\documentclass[english]{article}
\usepackage[T1]{fontenc}
\usepackage[latin9]{inputenc}
\usepackage{babel}
\usepackage{array}
\usepackage{float}
\usepackage{mathrsfs}
\usepackage{amsmath}
\usepackage{amssymb}
\usepackage{stmaryrd}
\usepackage{graphicx}
\usepackage[authoryear]{natbib}
\usepackage{nomencl}

\providecommand{\makenomenclature}{\makeglossary}
\makenomenclature
\usepackage[unicode=true,
 bookmarks=true,bookmarksnumbered=false,bookmarksopen=false,
 breaklinks=false,pdfborder={0 0 1},backref=false,colorlinks=false]
 {hyperref}
\hypersetup{pdftitle={\@shorttitle},
 pdfauthor={\@shortauthor},
 pdfsubject={\@Subject},
 pdfkeywords={Computer Graphics Forum, EUROGRAPHICS},
   colorlinks,linkcolor=blue,citecolor=blue,urlcolor=blue,    bookmarks=false,    pdfpagemode=UseNone}

\makeatletter

\providecommand{\tabularnewline}{\\}

\usepackage{arxiv}
\usepackage{inputenc}
\usepackage{tikz,tkz-tab,tkz-graph}
\usepackage{amsthm}
\usepackage{tikz}
\usepackage{dsfont}
\usepackage{stackengine}
\setstackgap{S}{1pt}
\usepackage{pgfplots}
\pgfplotsset{compat=newest}
\pgfplotsset{width=\textwidth,compat=1.8}
\pgfdeclarelayer{background}
\pgfsetlayers{background,main}

\usepgfplotslibrary{ternary}

\newtheorem{prop}{Proposition}

\newtheorem{lemma}{Lemma}

\theoremstyle{plain}

\usetikzlibrary{arrows,automata,shapes,snakes}
\usetikzlibrary{positioning,fadings,snakes,calc,fadings,shapes.arrows,shadows,backgrounds}
\usetikzlibrary{shapes.geometric}
\usetikzlibrary{cd}
\tikzset{state/.style={
	rectangle,
  	fill=#1!5!white,
    draw=#1, very thick,
    minimum height=2em,
    inner sep=2pt,
    text centered},
    statecirc/.style={
    circle,
    fill=#1!5!white,
    draw=#1, very thick,
    minimum height=2em,
    inner sep=2pt,
    text centered},
    highlight/.style={
		rectangle,
		fill=#1!50!white,rounded corners,draw=#1,very thick,
        minimum height=2em,
        inner sep=2pt,
        text centered},
	coeff/.style={
           circle,
           draw=black, very thick,
           minimum height=2em,
           inner sep=2pt,
           text centered},
	ptNode/.style={
		circle, 
		fill=black,thick, 
		inner sep=2pt, minimum size=0.2cm},
	square/.style={regular polygon sides=4},
	ptNodeSq/.style={square, fill=black,thick, inner sep=2pt, minimum size=0.2cm}
}

\usetikzlibrary{calc}

\tikzset{
	me/.style={
		to path={ 
		\pgfextra{%
		\pgfmathsetmacro{\startf}{-(#1-1)/2}    
		\pgfmathsetmacro{\endf}{-\startf}   
		\pgfmathsetmacro{\stepf}{\startf+1}}  
			\ifnum 1=#1 -- (\tikztotarget)  
			\else
				let \p{mid}=($(\tikztostart)!0.5!(\tikztotarget)$)           
				in	
				\foreach \i in {\startf,\stepf,...,\endf}     {%
				 (\tikztostart) .. controls ($ (\p{mid})!\i*6pt!90:(\tikztotarget) $) .. (\tikztotarget)       
				}       
			\fi         
		\tikztonodes     
		}
	}
} 

\tikzset{
    *|/.style={
        to path={
            (perpendicular cs: horizontal line through={(\tikztostart)},
                                 vertical line through={(\tikztotarget)})
            -- (\tikztotarget) \tikztonodes
        }
    }
}

\tikzset{
    invisible/.style={opacity=0},
    visible on/.style={alt=#1{}{invisible}},
    alt/.code args={<#1>#2#3}{%
      \alt<#1>{\pgfkeysalso{#2}}{\pgfkeysalso{#3}} 
    },
  }

\tikzset{fontscale/.style = {font=\relsize{#1}}
    }

\newcommand*\ghostPart{%
\startcontents
\phantomsection
\addcontentsline{toc}{part}{}%
\@endpart}

\usepackage{colortbl}
\definecolor{lightgray}{gray}{0.9}

\definecolor{S_purple}{RGB}{204, 0, 204}
\definecolor{S_brique}{RGB}{204, 51, 0}
\definecolor{S_petrol}{RGB}{0, 102, 153}
\definecolor{S_green}{RGB}{0, 153, 0}
\definecolor{darkorange}{RGB} {255,165,0}
\definecolor{yellowfluo}{RGB} {255,255,153}

\usepackage{ifthen}
\renewcommand{\nomgroup}[1]{%
\ifthenelse{\equal{#1}{B}}{\item[\textbf{Hb-graphs}]}{%
\ifthenelse{\equal{#1}{A}}{\item[\textbf{Hypergraphs}]}{%
\ifthenelse{\equal{#1}{T}}{\item[\textbf{Tensors}]}{%
\ifthenelse{\equal{#1}{X}}{\item[\textbf{General}]}
{}
}
}
}
}%

\newcommand*\circled[1]{\tikz[baseline=(char.base)]{
            \node[shape=circle,draw,inner sep=2pt] (char) {#1};}}

\usepackage{etoc, blindtext}

\newcommand*{\DotsAndPage}
{\nobreak\leaders\hbox{\bfseries\normalsize\hbox to .75ex {\hss.\hss}}%
\hfill\nobreak
\makebox[\rightskip][r]{\bfseries\normalsize\etocpage}\par}

\newcommand*{\DotsAndPageSec}
{\nobreak\leaders\hbox{\bfseries\small\hbox to .75ex {\hss.\hss}}%
\hfill\nobreak
\makebox[\rightskip][r]{\bfseries\small\etocpage}\par}

\newcommand*{\DotsAndPageSub}
{\nobreak\leaders\hbox{\footnotesize\sffamily\mdseries\itshape\hbox to .75ex {\hss.\hss}}%
\hfill\nobreak
\makebox[\rightskip][r]{\footnotesize{\sffamily\mdseries{\etocpage}}}\par}

\newcommand{\alltoc}[1]{%
\begingroup

\etocsetstyle {part}
{\parindent 0pt
\nobreak
\etocskipfirstprefix}
{\leftskip 0pt \rightskip .75cm \parfillskip-\rightskip
\pagebreak[1]\bigskip}
{\par\rule{\textwidth}{0.7pt}\par\normalsize\rmfamily\bfseries\scshape
\etocifnumbered{Part \etocnumber{} }{}\etocname\par\rule{\textwidth}{0.7pt}\par
}
{}

\etocsetstyle {chapter}
{\leftskip 0pt \rightskip .75cm \parfillskip-\rightskip
\nobreak\medskip
\etocskipfirstprefix}
{\leftskip 0pt \rightskip .75cm \parfillskip-\rightskip
\pagebreak[1]\smallskip}
{\normalsize\rmfamily\bfseries\scshape
\etocifnumbered{\etocnumber. }{}\etocname\DotsAndPage }
{\parfillskip 0pt plus 1fil\relax }

\etocsetstyle {section}
{\leftskip 0.5cm \rightskip .75cm \parfillskip-\rightskip
\nobreak\medskip
\etocskipfirstprefix}
{\leftskip 0.5cm \rightskip .75cm \parfillskip-\rightskip
\pagebreak[1]\smallskip}
{\small\rmfamily\bfseries\scshape
\etocnumber. \etocname\DotsAndPageSec }
{\parfillskip 0.5cm plus 1fil\relax }

\etocsetstyle {subsection}
{\leftskip1cm\rightskip 0.75cm \parfillskip-\rightskip
\nobreak\medskip
\etocskipfirstprefix}
{\leftskip 1cm \rightskip .75cm \parfillskip-\rightskip
\pagebreak[1]\smallskip}
{\footnotesize\sffamily\mdseries\itshape
{\etocnumber. }\etocname{} \DotsAndPageSub }
{\parfillskip 0.7cm \relax}

\etocsetstyle {subsubsection}
{\leftskip1.5cm\rightskip .75cm \parfillskip 0pt plus 1fil\relax
\nobreak\smallskip}
{\leftskip 1.5cm \rightskip .75cm \parfillskip-\rightskip
\pagebreak[1]\smallskip}
{\footnotesize\sffamily\mdseries
{\etocnumber. }\etocname{} \DotsAndPageSub }
{\parfillskip 0.7cm \relax}

\etocsettagdepth {preamble} {none}
\etocsettagdepth {linestyles} {subsection}
\etocsettagdepth {globalcmds} {subsection}
\etocsettagdepth {custom} {none}
\etocsettocstyle {\centering\LARGE\textsc{\contentsname}\par\nobreak\medskip}{}
\renewcommand{\baselinestretch}{0.75}\normalsize
\tableofcontents
\renewcommand{\baselinestretch}{1.0}\normalsize
\endgroup
}

\def\drawpolygon#1,#2,#3;{     
	\begin{pgfonlayer}{background}         		     
	\filldraw[line width=8,join=round,#1!30!white,opacity=0.6](#2)foreach\A in{#3}{--(\A)}--cycle;
	\filldraw[line width=5,join=round,#1!20!white,opacity=0.6](#2)foreach\A in{#3}{--(\A)}--cycle;     	\end{pgfonlayer} 
}

\def\drawpolygonwithborder#1,#2,#3;{     
	\begin{pgfonlayer}{background}         		     
	\filldraw[line width=20,join=round,#1!30!white](#2)foreach\A in{#3}{--(\A)}--cycle;     	 
	\filldraw[line width=18,join=round,#1!10!white](#2)foreach\A in{#3}{--(\A)}--cycle;     	
	\end{pgfonlayer} 
}

\usepackage{lscape}
\usepackage{adjustbox}

\usepgfplotslibrary{ternary}
\definecolor{blue}{rgb}{0.03, 0.35, 0.49}

\makeatother

\begin{document}
\title{Hypergraphs: an introduction and review}
\author{
	Xavier Ouvrard\thanks{xavier.ouvrard@cern.ch}\\
CERN,\\ 
Esplanade des Particules, 1,\\ 
CH-1211 Meyrin (Switzerland)
}
\maketitle
\begin{abstract}
Hypergraphs were introduced in 1973 by Berge. This review aims at giving some hints on the main results that we can find in the literature, both on the mathematical side and on their practical usage. Particularly, different definitions of hypergraphs are compared, some unpublished work on the visualisation of large hypergraphs done by the author. This review does not pretend to be exhaustive.
\end{abstract}

In 1736, Euler was the first to use a graph approach to solve the
problem of the Seven Bridges of Königsberg. In 1878, Sylvester coined
the word graph itself. Graphs are extensively used in various domains.
A lot of developments on graphs have been realized during the first
half of the twentieth century, but graphs are still an active field
of research and applications, in mathematics, physics, computer science
and various other fields, ranging from biology to economic and social
network analysis. 

In \citet{berge1967graphes,berge1973graphs}, the author introduces
hypergraphs as a means to generalize the graph approach. Graphs only
support pairwise relationships. Hypergraphs preserve the multi-adic
relationships and, therefore, become a natural modeling of collaboration
networks and various other situations. They already allow a huge step
in modeling, but some limitations remain that will be discussed further
in this article that has pushed us to introduce an extension of hypergraphs.
In this Section, we aim at giving a synthesis of the hypergraphs.
Further details on hypergraphs can be found either in \citet{berge1967graphes,berge1973graphs},
in \citet{voloshin2009introduction} or in \citet{bretto2013hypergraph}
and in many of the references cited throughout this article.

\section{Generalities}

\label{subsec:Background Hypergraph Generalities}

As given in \citet{berge1967graphes,berge1973graphs}, a \textbf{hypergraph}
$\ensuremath{\mathcal{H}=\left(V,E\right)}$\nomenclature[A,0001]{$\mathcal{H}=\left(V,E\right)$}{Hypergraph of vertex set V and hyperedge set E}\index{hypergraph}
on a finite \textbf{set of vertices}\index{vertex!hypergraph}\index{hypergraph!vertex}
(or nodes) $V\overset{\Delta}{=}\left\{ v_{i}:i\in\left\llbracket n\right\rrbracket \right\} $\footnote{We write for $n\in\mathbb{N}^{*}:$ $\left\llbracket n\right\rrbracket =\left\{ i:i\in\mathbb{N^{*}}\land i\leqslant n\right\} $}\nomenclature[X]{$\left\llbracket n\right\rrbracket$}{Integers from 1 to n}\nomenclature[A,0002]{$V=\left\{ v_{i}:i\in\left\llbracket n\right\rrbracket \right\}$}{Vertex set of a hypergraph $\mathcal{H}$}
is defined as a family of hyperedges\index{hyperedge!hypergraph}\index{hypergraph!hyperedge}
$E\overset{\Delta}{=}\left(e_{j}\right)_{j\in\left\llbracket p\right\rrbracket }$\nomenclature[A,0003]{$E=\left(e_{j}\right)_{j\in\left\llbracket p\right\rrbracket }$}{Hyperedge family of a hypergraph}
where each \textbf{hyperedge} is a non-empty subset of $V$ and such
that $\bigcup\limits _{j\in\left\llbracket p\right\rrbracket }e_{j}=V.$

It means that in a hypergraph, a hyperedge links one or more vertices.
In \citet{voloshin2009introduction}, the definition of hypergraphs
includes also hyperedges that are empty sets as hyperedges are defined
as a family of subsets of a finite vertex set and it is not necessary
that the union covers the vertex set. Both the vertex set and the
family of hyperedges can be empty; it they are at the same time, the
hypergraph is then designated as the \textbf{empty hypergraph}. This
definition of hypergraph opens their use in various collaboration
networks. It is the one we choose in this Thesis.

In \citet{bretto2013hypergraph}, an intermediate definition is taken,
relaxing only the covering of the vertex set by the union of the hyperedges
enabling only isolated vertices in hypergraphs.

Other interesting definitions of hypergraphs exist. Two of them are
given in \citet{stell2012relations}.

A hypergraph $\mathcal{H}$ is firstly defined as consisting of a
set $V$ of vertices and a set $E$ of hyperedges with an \textbf{incidence
function}\index{incidence function!hypergraph}\index{hypergraph!incidence function}
$\iota\,:\,E\rightarrow\mathcal{P}\left(V\right),$ where $\mathcal{P}\left(V\right)$
is the poset of $V$. Also $\mathcal{H}=\left(V,E,\iota\right).$\nomenclature[A,0004]{$\mathcal{H}=\left(V,E,\iota\right)$}{Hypergraph of vertex set V, hyperedge set E and incidence relation $\iota$}

The author emphasizes the fact that this definition allows several
hyperedges to be incident to the same set of vertices and, also, to
have hyperedges that are linked to the empty set of vertices. This
definition is tightly linked to the one of \citet{bretto2013hypergraph},
with the advantage of having the incident function to link hyperedges
and their corresponding set of vertices, instead of an equality.

A hypergraph is secondly defined in \citet{stell2012relations} as
consisting of a set $\mathcal{H}$ containing both vertices and hyperedges
and a \textbf{binary relation} $\varphi$\index{binary relation!hypergraph}\index{hypergraph!binary relation}
such that:
\begin{enumerate}
\item $\forall x\in\mathcal{H},\,\forall y\in\mathcal{H},\,\,\,\,(x,y)\in\varphi\implies(y,y)\in\varphi;$
\item $\forall x\in\mathcal{H},\,\forall y\in\mathcal{H},\,\forall z\in\mathcal{H},\,\,\,(x,y)\in\varphi\land(y,z)\in\varphi\implies y=z.$
\end{enumerate}
The link can be made with the other definitions in considering the
set $V$ such that:
\[
V\overset{\Delta}{=}\left\{ v\in\mathcal{H}:\left(v,v\right)\in\varphi\right\} 
\]
which corresponds to the set of vertices of the hypergraph $\mathcal{H}$
and the set $E$ such that:
\[
E\overset{\Delta}{=}\left\{ e\in\mathcal{H}:\,\left(e,e\right)\notin\varphi\right\} 
\]
which corresponds to the set of hyperedges of the hypergraph.

This last definition treats on the same footing vertices and hyperedges.
Nonetheless, in this definition, the set aspect of hyperedges is implicit.
Effectively, to know which vertex set a hyperedge $e$ contains, the
information needs to be rebuilt---using $\iota(e)$ as in the former
definition---: 
\[
\iota(e)=\left\{ v\in\mathcal{H}:\,\left(e,v\right)\in\varphi\right\} .
\]

This definition leads to a representation of the hypergraph as a directed
multi-graph where the vertices point to themselves and hyperedges
are linked to vertices in this order.

Depending on the needs, either the first definition of \citet{stell2012relations}
or equivalently the one given in \citet{bretto2013hypergraph} without
specifying the incidence relation will be used and referred as poset
definition of hypergraphs---\textbf{PoDef}\index{PoDef}\index{hypergraph!PoDef}
for short. Switching between these two definitions reverts to abusively
identifying the hyperedge with the subset of vertices it identifies.
The second definition of \citet{stell2012relations} will also be
used---referred to as \textbf{MuRelDef}\index{MuRelDef}\index{hypergraph!MuRelDef}.

When using a PoDef defined hypergraph, $\mathcal{H}=\left(V,E,\iota\right)$
designates a hypergraph. If the MuRelDef is used, the hypergraph $\mathcal{H}=\left(U,\varphi\right)$
is used, with $U=V\cup E,$ with $V=\left\{ v\in\mathcal{H}:\left(v,v\right)\in\varphi\right\} $
the set of vertices and $E=\left\{ e\in\mathcal{H}:\,\left(e,e\right)\notin\varphi\right\} .$
In the following definitions, we put in perspective these two approaches.

\section{Particular hypergraphs}

When needed, elements of $V$ will be written $\left(v_{i}\right)_{i\in\left\llbracket n\right\rrbracket }$
and those of $E$ will be written as $\left(e_{j}\right)_{j\in\left\llbracket p\right\rrbracket }.$
Abusively, $e_{j}$ will also designate the subset $\iota\left(e_{j}\right)$
of $V.$

Two hyperedges $e_{j_{1}}\in E$ and $e_{j_{2}}\in E$, with $j_{1},j_{2}\in\left\llbracket p\right\rrbracket $
and $j_{1}\neq j_{2}$ such that $e_{j_{1}}=e_{j_{2}}$ are said \textbf{repeated
hyperedges}\index{repeated hyperedges}\index{hypergraph!repeated hyperedges}\index{hyperedge!repeated}.

A hypergraph is said with \textbf{no repeated hyperedge}\index{no repeated!hyperedge}\index{hypergraph!no repeated hyperedge}\index{hyperedge!no repeated},
if it has no repeated hyperedges.

Following \citet{chazelle1988deterministic}, where the hyperedge
family is viewed as a multiset of hyperedges, a hypergraph with repeated
hyperedges is called a \textbf{multi-hypergraph}.\index{multi-hypergraph}\index{hypergraph!multi-hypergraph}

A hypergraph is said \textbf{simple}\index{hypergraph!simple}\index{simple!hypergraph},
if for any two hyperedges $e_{j_{1}}\in E$ and $e_{j_{2}}\in E:$
\[
e_{j_{1}}\subseteq e_{j_{2}}\Rightarrow j_{1}=j_{2}.
\]
Hence, a simple hypergraph has no repeated hyperedges.

A hyperedge $e\in E$ such that: $\left|e\right|=1$ is called a \textbf{loop}.

A hypergraph is said \textbf{linear}\index{hypergraph!linear}\index{linear!hypergraph}
if it is simple and such that every pair of hyperedges shares at most
one vertex.

A \textbf{sub-hypergraph}\index{sub-hypergraph}\index{hypergraph!sub-hypergraph}
$\mathcal{K}$ of a hypergraph $\mathcal{H}$ is the hypergraph formed
by a subset $W$ of the vertices of $\mathcal{H}$ and the hyperedges
of $\mathcal{H}$ that are subsets of $W$.

Formally, with 
\begin{itemize}
\item PoDef: $\mathcal{K}=\left(W,F,\iota|_{F}\right)$, such that: $F=\left\{ e_{j}\in E\,:\,\iota|_{F}\left(e_{j}\right)\subseteq W\right\} ;$
\item MuRelDef: $\mathcal{K}=\left(T,\varphi\right)$, with $T\subseteq U$,
where $\forall t\in T,\,\forall u\in U:\,\left(t,u\right)\in\varphi\implies u\in T$
and $\forall t\in T:\,\left(t,t\right)\in\varphi\implies t\in W.$
\end{itemize}
A \textbf{partial hypergraph}\index{partial hypergraph}\index{hypergraph!partial hypergraph}
$\mathcal{H}'$ generated by a subset $E'\subseteq E$ of the hyperedges
is a hypergraph containing exactly these hyperedges and whose vertex
set contains at least all the vertices incident to this hyperedge
subset.

Formally, with:
\begin{itemize}
\item PoDef: $\mathcal{H}'\overset{\Delta}{=}\left(V',E',\iota|_{E'}\right)$,
where $E'=\left\{ e_{j}\,:\,j\in K'\right\} \subseteq E$ with $K'\subseteq\left\llbracket p\right\rrbracket $
and $V'$ is such that $\bigcup\limits _{k\in K'}\iota|_{E'}\left(e_{k}\right)\subseteq V';$
\item MuRelDef: $\mathcal{H}'\overset{\Delta}{=}\left(T',\varphi\right)$,
with $T'\subseteq U$, where $\forall t\in U,\forall t'\in T':\,\left(t,t'\right)\in\varphi\implies t\in T'$
and $\forall t\in T':\,\left(t,t\right)\notin\varphi\implies t\in E'.$
\end{itemize}

\section{Duality}

The \textbf{star of a vertex $v\in V$} of a hypergraph $\mathcal{H}$,
written $H(v)$,\index{star}\index{hypergraph!vertex!star}\index{vertex!hypergraph!star}\nomenclature[A,0005]{$H(v)$}{Star of a vertex $v \in V$}
is the family of hyperedges containing $v.$

Also with:
\begin{itemize}
\item PoDef: $H(v)\overset{\Delta}{=}\left\{ e_{k}\in E:\,v\in\iota\left(e_{k}\right)\right\} ;$
\item MuRelDef: $H(v)\overset{\Delta}{=}\left\{ t\in U:\,\left(t,v\right)\in\varphi\land\left(t,t\right)\notin\varphi\right\} .$
\end{itemize}
The \textbf{dual}\index{dual!hypergraph}\index{hypergraph!dual}
of a hypergraph $\mathcal{H}$ is the hypergraph $\mathcal{H}^{*}$\nomenclature[A,0006]{$\mathcal{H}^{*}$}{Dual of a hypergraph $\mathcal{H}$}
whose vertices corresponds to the hyperedges of $\mathcal{H}$ and
the hyperedges of $\mathcal{H}^{*}$ are the vertices of $\mathcal{H},$
with the incident relation that links each vertex to its star.

Formally with:
\begin{itemize}
\item PoDef: The dual is given by $\mathcal{H}^{*}=\left(V^{*},E^{*},\iota^{*}\right),$
with $V^{*}=E$ and where $E^{*}=\left\{ v\in V\right\} $ and $\forall v\in V:\,\iota^{*}\left(v\right)=H\left(v\right);$
\item MuRelDef: The dual is given by: $\left(U,\varphi^{*}\right)$ where
$\varphi^{*}\overset{\Delta}{=}\lnot\varphi\cap\left(1^{\prime}\cup\smallsmile\varphi\right),$
with:
\item $\lnot\varphi$ the complement of the relation $\varphi:$ $\forall x\in U,\forall y\in U:\,\,\,\left(x,y\right)\in\lnot\varphi\Leftrightarrow(x,y)\notin\varphi;$
\item 1' the binary identity: $\forall x\in U,\forall y\in U:\,\,\,(x,y)\in1^{\prime}\Leftrightarrow x=y;$
\item $\smallsmile\varphi$ the converse relation: $\forall x\in U,\forall y\in U:\,\,\,\left(x,y\right)\in\smallsmile\varphi\Leftrightarrow(y,x)\in\varphi.$
\end{itemize}
\begin{proof}[Proof for (MuRelDef)](not given in \citet{stell2012relations})
Let $x\in U$ such that $\left(x,x\right)\in\varphi^{*}.$ Then $\left(x,x\right)\in\lnot\varphi$
and $x$ is a hyperedge of $\left(U,\varphi\right).$

If $x\in U$ such that $\left(x,x\right)\notin\varphi^{*},$ then
$\left(x,x\right)\notin\lnot\varphi$ and $x$ is a vertex of $\left(U,\varphi\right)$
as $\left(x,x\right)\in1^{\prime}$ and, therefore: $(x,x)\in1'\cup\smallsmile\varphi.$

Let $\left(x,y\right)\in\varphi^{*}$ and suppose that $x\neq y.$
Then as $\varphi^{*}=\lnot\varphi\cap\left(1^{\prime}\cup\smallsmile\varphi\right),$
we have: $\left(x,y\right)\in1^{\prime}\cup\smallsmile\varphi.$ As
$x\neq y,$ $\left(x,y\right)\notin1^{\prime},$ so we have $\left(x,y\right)\in\smallsmile\varphi,$
i.e.\ $\left(y,x\right)\in\varphi.$ It implies that $\left(x,x\right)\in\varphi$
and thus a vertex of $\left(U,\varphi\right).$ As $y$ is a vertex
of $\left(U,\varphi^{*}\right),$ it is a hyperedge of $\left(U,\varphi\right).$

\end{proof}

Some interesting properties can be derived from the MuRelDef and might
be useful for computation.

\begin{prop}

Let $\left(U,\varphi\right)$ be a hypergraph. 

The set of vertices is given by the binary relation: $\nu\overset{\Delta}{=}\varphi\cap1'.$

The set of hyperedges is given by the binary relation: $\eta\overset{\Delta}{=}\varphi\cap\left(\lnot\nu\right).$

\end{prop}

\begin{proof}

Let $x\in U$: $\left(x,x\right)\in\varphi\Leftrightarrow\left(x,x\right)\in\varphi\cap1',$
so $\nu$ gives exactly the vertices of $\left(U,\varphi\right).$

Moreover: $\left(x,x\right)\in\varphi\Leftrightarrow\left(x,x\right)\notin\lnot\nu,$
so $\eta=\varphi\cap\left(\lnot\nu\right)$ does not hold the vertices
of $\left(U,\varphi\right).$

Let us take $x,y\in U$ with $x\neq y.$ We have $\left(x,y\right)\in\varphi\implies\left(x,y\right)\notin\varphi\cap1'\Leftrightarrow\left(x,y\right)\in\lnot\left(\varphi\cap1'\right).$
Therefore: $\left(x,y\right)\in\varphi\implies\left(x,y\right)\in\varphi\cap\left(\lnot\nu\right).$ 

Reciprocally, if $x\neq y:$ $\left(x,y\right)\in\varphi\cap\left(\lnot\nu\right)\implies\left(x,y\right)\in\varphi,$
then $\eta$ gives exactly the hyperedges' content of $\left(U,\varphi\right).$\end{proof}

\begin{prop}

Using the dual relation:

The set of hyperedges can be retrieved by the binary relation: $\nu^{*}\overset{\Delta}{=}\varphi^{*}\cap1'.$

The set of incident hyperedges can be retrieved by the binary relation:
$\eta^{*}\overset{\Delta}{=}\varphi^{*}\cap\left(\lnot\nu^{*}\right).$

\end{prop}

\begin{proof}Immediate with duality.\end{proof}

\begin{prop}

Let $\left(U,\varphi\right)$ be a hypergraph and $\varphi^{*}$ its
dual relation.

(i) The binary relation $\alpha=\left(\varphi\circ\varphi^{*}\right)\cap\left(\lnot\eta\right)$
allows to retrieve the adjacency of vertices.

(ii) The binary relation $\beta=\left(\varphi^{*}\circ\varphi\right)\cap\left(\lnot\eta^{*}\right)$
allows to retrieve the intersection of hyperedges.

\end{prop}

\begin{proof}[Proof for (i)]In $\left(U,\varphi^{*}\right),$ when
we select an element of $x\in U$, either it is a vertex of $\left(U,\varphi^{*}\right),$
and then $\left(x,x\right)\in\varphi^{*},$ or it is a hyperedge and
then $\left(x,x\right)\notin\varphi^{*}.$

Let $x\in\left(U,\varphi^{*}\right)$ be such that $\left(x,x\right)\notin\varphi^{*}.$
$x$ is a hyperedge of $\left(U,\varphi^{*}\right)$ and, therefore,
a vertex of $\left(U,\varphi\right).$ If there exists $e\in U,$
such that $\left(x,e\right)\in\varphi^{*},$ then $e$ is a vertex
of $\left(U,\varphi^{*}\right),$ also a hyperedge of $\left(U,\varphi\right).$
If this hyperedge is non empty in $\left(U,\varphi\right),$ there
exists $x'\in\left(U,\varphi\right),$ such that $\left(e,x'\right)\in\varphi,$
and, necessarily by definition, $\left(x',x'\right)\in\varphi.$ $\left(x,x'\right)$
links two vertices of $\left(U,\varphi\right).$ Either $x=x'$ which
implies: $\left(x,x'\right)\in\lnot\eta$ or $x\neq x'$ which implies
$\left(x,x'\right)\notin\varphi,$ so $\left(x,x'\right)\in\lnot\eta.$
Therefore, $\left(x,x'\right)\in\left(\varphi\circ\varphi^{*}\right)\cap\left(\lnot\eta\right)$
links two vertices and points the adjacency of these two vertices.

Let $x\in\left(U,\varphi^{*}\right)$ be such that $\left(x,x\right)\in\varphi^{*}.$
$x$ is a vertex of $\left(U,\varphi^{*}\right)$ and, therefore,
a hyperedge of $\left(U,\varphi\right),$ i.e.\ $\left(x,x\right)\notin\varphi$.
If this hyperedge is non empty in $\left(U,\varphi\right),$ there
exists $x'\in\left(U,\varphi\right),$ such that $\left(x,x'\right)\in\varphi,$
and necessarily, by definition, $\left(x',x'\right)\in\varphi,$ and
$x'$ is a vertex of $\left(U,\varphi\right).$ Therefore, $\left(x,x'\right)\in\left(\varphi\circ\varphi^{*}\right)\cap\eta.$\end{proof}

\begin{proof}[Proof for (ii)]This proof is similar to the former
one.\end{proof}

\section{Neighborhood of a vertex}

The neighborhood of a vertex in a hypergraph is defined similarly
to what is done for graphs. 

The \textbf{neighborhood of a vertex} $v\in V$\index{neighborhood!hypergraph}\index{hypergraph!vertex!neighborhood}\index{vertex!hypergraph!neighborhood}
is the set $\Gamma\left(v\right)$\nomenclature[A,0007]{$\Gamma\left(v\right)$}{Neighbourhood of a vertex $v \in V$}
of vertices that belongs to the hyperedges this vertex is belonging.

That is with:
\begin{itemize}
\item PoDef: $\Gamma(v)=\bigcup\limits _{e\in H(v)}\iota(e);$
\item MuRelDef: For $v$ such that $\left(v,v\right)\in\varphi:$
\[
\Gamma(v)=\left\{ s\in U:\,\left(s,s\right)\in\varphi\land\left(\exists w\in U:\left(w,s\right)\in\varphi\land\left(w,v\right)\in\varphi\land\left(w,w\right)\notin\varphi\right)\right\} ,
\]
which is equivalent to: 
\[
\Gamma(v)=\left\{ s\in U:\,\left(s,s\right)\in\varphi\land\left(s,v\right)\in\varphi\circ\varphi^{*}\right\} .
\]
\end{itemize}
\begin{proof}[Proof for (MuRelDef)]As $\left(v,v\right)\in\varphi,$
$\left(v,v\right)\notin\varphi^{*}.$ If there exists $t\in U$ such
that $(t,v)\in\varphi^{*},$ $t$ is a hyperedge of $\left(U,\varphi\right).$
Therefore, any $s$, if there exists, such that $\left(s,t\right)\in\varphi$
is a vertex and is in the same hyperedge than $v,$ so in its neighborhood.
The other inclusion is straightforward.\end{proof}

\section{Weighted hypergraph}

In \citet{zhou2007learning}, the definition of a weighted hypergraph
is given, based on the definition of \citet{berge1967graphes,berge1973graphs}
of a hypergraph.

$\mathcal{H}_{w_{e}}=\left(V,E,w_{e}\right)$ is a \textbf{weighted
hypergraph}\index{hypergraph!weighted}\index{weighted hypergraph}\nomenclature[A,0009.5]{$\mathcal{H}_{w}=\left(V,E,w\right)$}{Weighted hypergraph of vertex set $V$, edge family $E$ and weight function $w$},
if $\left(V,E\right)$ is a hypergraph and $w_{e}:E\rightarrow\mathbb{R}$
is a function that associates to each hyperedge $e\in E$ a weight
$w_{e}\left(e\right).$

We can refine this definition to handle weights on individual vertices,
by using a second function $w_{v}:V\rightarrow\mathbb{R}$ that associates
to each vertex $v\in V$ a weight $w_{v}\left(v\right).$ But putting
weights that are hyperedge dependent cannot be achieved with hypergraphs
as it would imply to move to a new algebra, as we will see with the
introduction of hb-graphs.

\section{Hypergraph features}

\label{subsec:Hypergraph-features}

Hypergraph features are very similar to those of graphs with some
arrangements to account for their differences in structure.

The \textbf{order} $o_{\mathcal{H}}$\index{hypergraph!order}\index{order!hypergraph}
\nomenclature[A,0010]{$o_\mathcal{H}$}{Order of a hypergraph $\mathcal{H}$}
of a hypergraph $\mathcal{H}$ is defined as $o_{\mathcal{H}}\overset{\Delta}{=}\left|V\right|.$

The \textbf{rank} $r_{\mathcal{H}}$\index{hypergraph!rank}\index{rank!hypergraph}\nomenclature[A,0011]{$r_\mathcal{H}$}{Rank of a hypergraph $\mathcal{H}$}
of a hypergraph $\mathcal{H}$ is the maximum of the cardinalities
of the hyperedges:
\[
r_{\mathcal{H}}\overset{\Delta}{=}\underset{e\in E}{\max}\left|e\right|,
\]
while the \textbf{anti-rank} $s_{\mathcal{H}}$\index{hypergraph!anti-rank}\index{anti-rank!hypergraph}\nomenclature[A,0012]{$s_\mathcal{H}$}{Anti-rank of a hypergraph $\mathcal{H}$}
corresponds to the minimum: 
\[
s_{\mathcal{H}}\overset{\Delta}{=}\underset{e\in E}{\min}\left|e\right|.
\]

The \textbf{degree of a vertex $v_{i}\in V$}\index{hypergraph!vertex!degree}\index{degree!vertex!hypergraph}\index{vertex!hypergraph!degree},
written \textbf{$\text{deg}\left(v_{i}\right)=d_{i},$} corresponds
to the number of hyperedges that this vertex participates in. Hence:
\[
\text{deg}\left(v_{i}\right)\overset{\Delta}{=}\left|H\left(v_{i}\right)\right|.
\]
It is also designated as \textbf{hyper-degree}\index{hypergraph!vertex!hyper-degree}\index{hyper-degree!vertex!hypergraph}\index{vertex!hypergraph!hyper-degree}
in some articles.

\section{Paths and related notions}

A \textbf{path}\index{path!hypergraph}\index{hypergraph!path} $v_{i_{0}}e_{j_{1}}v_{i_{1}}\ldots e_{j_{s}}v_{i_{s}}$
in a hypergraph $\mathcal{H}$ from a vertex $u$ to a vertex $v$
is a vertex / hyperedge alternation with $s$ hyperedges $e_{j_{k}}$
such that: $\forall k\in\left\llbracket s\right\rrbracket ,j_{k}\in\left\llbracket p\right\rrbracket $
and $s+1$ vertices $v_{i_{k}}$ with $\forall k\in\left\{ 0\right\} \cup\left\llbracket s\right\rrbracket ,\,i_{k}\in\left\llbracket n\right\rrbracket $
and such that $v_{i_{0}}=u$, $v_{i_{s}}=v$, $u\in e_{j_{1}}$ and
$v\in e_{j_{s}}$ and that for all $k\in\left\llbracket s-1\right\rrbracket $,
$v_{i_{k}}\in e_{j_{k}}\cap e_{j_{k+1}}.$ 

The \textbf{length of a path}\index{hypergraph!path!length}\index{path!length}\index{length!path}
from $u$ to $v$ is the number of hyperedges it traverses; given
a path $\mathscr{P}$, we write $l\left(\mathscr{P}\right)$\nomenclature[A,008]{$l\left(\mathscr{P}\right)$}{Length of a path $\mathscr{P}$}
its length. It holds that if $\mathscr{P}=v_{i_{0}}e_{j_{1}}v_{i_{1}}\ldots e_{j_{s}}v_{i_{s}},$
we have: $l\left(\mathscr{P}\right)=s.$

The \textbf{hypergraph distance}\index{hypergraph!hypergraph distance}\index{hypergraph distance!hypergraph}
$d\left(u,v\right)$ \nomenclature[A,009]{$d\left(u,v\right)$}{Distance between two vertices $u$ and $v$}
between two vertices $u$ and $v$ of a hypergraph is the length of
the shortest path between $u$ and $v$, if there exists, that can
be found in the hypergraph. In the case where there is no path between
the two vertices, they are said to be \textbf{disconnected}, and we
set: $d\left(u,v\right)=+\infty.$ A hypergraph is said to be \textbf{connected}\index{hypergraph!connected}\index{connected hypergraph}
if there exists a path between every pair of vertices of the hypergraph,
and \textbf{disconnected} otherwise.

A \textbf{connected component}\index{hypergraph!connected component}\index{component!hypergraph!connected}\index{connected component}
of a hypergraph is a maximal subset of the vertex set for which there
exists a path between any two vertices of this maximal subset in the
hypergraph.

\section{Multi-graph, graph, 2-section}

A hypergraph with rank at most 2 is called a \textbf{multi-graph}\index{multi-graph}.
A simple multi-graph without loop is a \textbf{graph}\index{graph}.

For the moment, we keep the original concept of adjacency as it is
implicitly given in \citet{bretto2013hypergraph}; we mention it here
as \textbf{2-adjacency} since it is a pairwise adjacency.

Two distinct vertices $v_{i_{1}}\in V$ and $v_{i_{2}}\in V$ are
said \textbf{2-adjacent}\index{2-adjacent} if there exists $e\in E$
such that $v_{i_{1}}\in e$ and $v_{i_{2}}\in e.$

The graph $\left[\mathcal{H}\right]_{2}\overset{\Delta}{=}\left(V_{\left[2\right]},E_{\left[2\right]}\right)$\nomenclature[A,0012.2]{$\left[\mathcal{H}\right]_{2}$}{2-section of the hypergraph $\mathcal{H}$}
obtained from a hypergraph $\mathcal{H}=\left(V,E\right)$ by considering:
$V_{\left[2\right]}\overset{\Delta}{=}V$ and such that if $v_{i_{1}}$
and $v_{i_{2}}$ are 2-adjacent in $\mathcal{H}$, $\left\{ v_{i_{1}},v_{i_{2}}\right\} \in E_{[2]}$
is called the \textbf{2-section of the hypergraph} $\mathcal{H}.$\index{hypergraph!2-section}\index{2-section of a hypergraph}

The graph $\left[\mathcal{H}\right]_{\mathcal{I}}\overset{\Delta}{=}\left(V_{\mathcal{I}},E_{\mathcal{I}}\right)$\nomenclature[A,0012.5]{$\left[\mathcal{H}\right]_{\mathcal{I}}$}{Interection graph of the hypergraph $\mathcal{H}$}
obtained from a hypergraph $\mathcal{H}=\left(V,E\right)$ by considering:
$V_{\mathcal{I}}\overset{\Delta}{=}V^{*}$ and, such that, if $e_{j_{1}}\in E$
and $e_{j_{2}}\in E$---with $j_{1}\neq j_{2}$---are intersecting
hyperedges in $\mathcal{H},$ then $\left\{ e_{j_{1}},e_{j_{2}}\right\} \in E_{\mathcal{I}},$
is called the \textbf{intersection graph} of the hypergraph $\mathcal{H}.$
\index{hypergraph!intersection graph}\index{intersection graph!hypergraph}

Let $k\in\mathbb{N}^{*}.$ A hypergraph is said to be \textbf{$\boldsymbol{k}$-uniform}\index{hypergraph!k-uniform@$k\textrm{-uniform}$}\index{k-uniform@$k\textrm{-uniform}$!hypergraph}
if all its hyperedges have the same cardinality $k.$

A \textbf{directed hypergraph} $\mathcal{H}=\left(V,E\right)$\index{hypergraph!directed}\index{directed!hypergraph}
on a finite set of $n$ vertices (or vertices) $V=\left\{ v_{i}:i\in\left\llbracket n\right\rrbracket \right\} $
is defined as a family of $p$ \textbf{hyperedges} $E=\left(e_{j}\right)_{j\in\left\llbracket p\right\rrbracket }$
where each hyperedge $e_{j}$ contains exactly two non-empty subsets
of $V$, one which is called the \textbf{source}---written $e_{s\,j}$---and
the other one which is the \textbf{target}---written $e_{t\,j}.$
A hypergraph that is not directed is said to be an \textbf{undirected
hypergraph}.

\section{Sum of hypergraphs}

Let $\mathcal{H}_{1}=\left(V_{1},E_{1}\right)$ and $\mathcal{H}_{2}=\left(V_{2},E_{2}\right)$
be two hypergraphs. The \textbf{sum} of\index{sum!hypergraph}\index{hypergraph!sum}
these two hypergraphs is the hypergraph written $\mathcal{H}_{1}+\mathcal{H}_{2}$
defined as: 
\[
\mathcal{H}_{1}+\mathcal{H}_{2}\overset{\Delta}{=}\left(V_{1}\cup V_{2},E_{1}\cup E_{2}\right).
\]
\nomenclature[A,0012.7]{$\mathcal{H}_1+\mathcal{H}_2$}{Sum of two hypergraphs $\mathcal{H}_1$ and $\mathcal{H}_2$}This
sum is said \textbf{direct} if $E_{1}\cap E_{2}=\emptyset.$ In this
case, the sum is written $\mathcal{H}_{1}\oplus\mathcal{H}_{2}$\nomenclature[A,0012.8]{$\mathcal{H}_1\oplus\mathcal{H}_2$}{Direct sum of two hypergraphs $\mathcal{H}_1$ and $\mathcal{H}_2$}\index{direct sum!hypergraph}\index{hypergraph!direct sum}.

\section{Matrix definitions of hypergraphs}

\subsubsection{Incidence matrix}

The \textbf{incidence matrix}\index{hypergraph!incidence matrix}\index{incidence matrix!hypergraph}\nomenclature[A,0013]{$H$}{Incidence matrix of a hypergraph}:
\[
H\overset{\Delta}{=}\left[h_{ij}\right]_{\substack{i\in\left\llbracket n\right\rrbracket \\
j\in\left\llbracket p\right\rrbracket 
}
}
\]
of a hypergraph is the matrix having rows indexed by the corresponding
indices of vertices of $\mathcal{H}$ and columns by the hyperedge
indices, and where the coefficient $h_{ij}\overset{\Delta}{=}1$ when
$v_{i}\in e_{j},$ and $h_{ij}\overset{\Delta}{=}0$ when $v_{i}\notin e_{j}.$

\subsubsection{Adjacency matrix}

\label{subsec:Adjacency-matrix}

We focus in this paragraph on the pairwise adjacency as defined in
\citet{berge1967graphes,berge1973graphs} and \citet{bretto2013hypergraph}.
In the latter, the \textbf{adjacency matrix of a hypergraph} $\mathcal{H}$\index{hypergraph!adjacency matrix}\index{adjacency matrix!unweighted hypergraph}
is defined as the square matrix $A\overset{\Delta}{=}\left[a_{ij}\right]$\nomenclature[A,0014]{$A$}{Adjacency matrix of a hypergraph}
having rows and columns indexed by indices corresponding to the indices
of vertices of $\mathcal{H}$ and where for all $i,j\in\left\llbracket p\right\rrbracket ,$
$i\neq j:$ 
\[
a_{ij}\overset{\Delta}{=}\left|\left\{ e\in E\,:\,v_{i}\in e\land v_{j}\in e\right\} \right|
\]
 and for all $i\in\left\llbracket p\right\rrbracket :$ 
\[
a_{ii}\overset{\Delta}{=}0.
\]

It holds, following \citet{estrada2005complex}:
\[
A=HH^{\intercal}-D_{V}
\]
where $D_{V}\overset{\Delta}{=}\text{diag}\left(\left(d_{i}\right)_{i\in\left\llbracket n\right\rrbracket }\right)$
is the diagonal matrix containing vertex degrees.

The \textbf{adjacency matrix of a weighted hypergraph}\index{hypergraph!adjacency matrix of a weighted}\index{adjacency matrix!weighted hypergraph}
$\mathcal{H}_{w}$ is defined in \citet{zhou2007learning} as the
matrix $A_{w}$\nomenclature[A,0015]{$A_w$}{Adjacency matrix of a weighted hypergraph}
of size $n\times n$ defined as:
\[
A_{w}\overset{\Delta}{=}HWH^{\intercal}-D_{w}
\]
where:
\[
W\overset{\Delta}{=}\text{diag}\left(\left(w_{j}\right)_{j\in\left\llbracket p\right\rrbracket }\right)
\]
is a diagonal matrix of size $p\times p$ containing the weights $\left(w_{j}\right)_{j\in\left\llbracket p\right\rrbracket }$
of the respective hyperedges $\left(e_{j}\right)_{j\in\left\llbracket p\right\rrbracket }$
and $D_{w}$ is the diagonal matrix of size $n\times n$:
\[
D_{w}\overset{\Delta}{=}\text{diag}\left(\left(d_{w}\left(v_{i}\right)\right)_{i\in\left\llbracket n\right\rrbracket }\right)
\]
and where for all $i\in\left\llbracket n\right\rrbracket $: 
\[
d_{w}\left(v_{i}\right)\overset{\Delta}{=}\sum\limits _{j\in\left\{ k:k\in\left\llbracket p\right\rrbracket \land v_{i}\in e_{k}\right\} }w_{j}
\]
is the \textbf{weighted degree} of the vertex $v_{i}\in V.$

We will refine the concept of adjacency in general hypergraphs in
the next Chapter.

\subsubsection{Additional features}

In \citet{estrada2005complex}, the authors introduce particular features
to characterize hypergraphs. They evaluate the \textbf{centrality
of a vertex}\index{vertex!hypergraph!centrality}\index{centrality!vertex}\index{vertex!hypergraph!centrality}
in a simple unweighted hypergraph, by orthogonalizing the adjacency
matrix in:
\[
A=U\Lambda U^{\intercal},
\]
where $U\overset{\Delta}{=}\left(u_{ij}\right)_{\left(i,j\right)\in\left\llbracket n\right\rrbracket ^{2}}$
and $\Lambda\overset{\Delta}{=}\text{diag}\left(\left(\lambda_{i}\right)_{i\in\left\llbracket n\right\rrbracket }\right)$
is the diagonal matrix formed of the eigenvalues $\lambda_{i}$ of
$A$ with $i\in\left\llbracket n\right\rrbracket .$

The \textbf{sub-hypergraph centrality}\index{sub-hypergraph centrality}\index{centrality!sub-hypergraph}\index{hypergraph!sub-hypergraph!centrality}
$C_{\text{SH}}(i)$ is defined as the sum of the closed walks of different
lengths in the network, starting and ending at a given vertex. For
a simple hypergraph, it can be calculated using: 
\[
C_{\text{SH}}(i)\overset{\Delta}{=}\sum\limits _{j=1}^{n}\left(u_{ij}\right)^{2}\mbox{e}^{\lambda_{j}}.
\]
A \textbf{clustering coefficient}\index{hypergraph!clustering coefficient}\index{clustering coefficient!hypergraph}
for hypergraph is also defined as: 
\[
C(\mathcal{H})\overset{\Delta}{=}\dfrac{6\times\mbox{number of hyper-triangles}}{\mbox{number of 2-paths}}
\]
where a hyper-triangle is defined as a sequence of three different
vertices that are separated by three different hyperedges $v_{i}e_{p}v_{j}e_{q}v_{k}e_{r}v_{i}$
and a 2-path is a sequence $v_{i}e_{p}v_{j}e_{q}v_{k}.$

\section{Hypergraph visualisation}

\label{subsec:Hypergraph-visualisation}

In \citet{makinen1990how}, hypergraph visualizations are classified
in two categories called the ``subset standard'' and the ``edge
standard''. These two types of representations reflect the two facets
of hypergraphs. The \textbf{subset standard}\index{hypergraph!representation!subset standard}\index{subset standard!hypergraph}\index{representation!hypergraph!subset standard}
reflects that hyperedges are subsets of the vertex set: the vertices
of a hyperedge are drawn as points and hyperedges as closed envelopes.
The \textbf{edge standard}\index{hypergraph!representation!edge standard}\index{edge standard!hypergraph}\index{representation!hypergraph!edge standard}
reflects that vertices of a hyperedge maintain together a multi-adic
relationship. We then cover the Zykov representation which can be
seen as a bridge between the two representations.

\subsubsection{The subset standard}

\textbf{Euler diagram}\index{Euler diagram} represents sets by simple
closed shapes in a two dimensional plane. Shapes can overlap to represent
relationships within the sets; these relationships can be either fully
inclusive, partially inclusive or disjoint.

In \citet{venn1881diagrammatic}, the author introduces a particular
case of Euler diagram that shows the $2^{n}$ logical relations between
$n$ sets with the aim to formalize logical propositions by diagrams,
now called \textbf{Venn diagrams}\index{Venn diagram}.

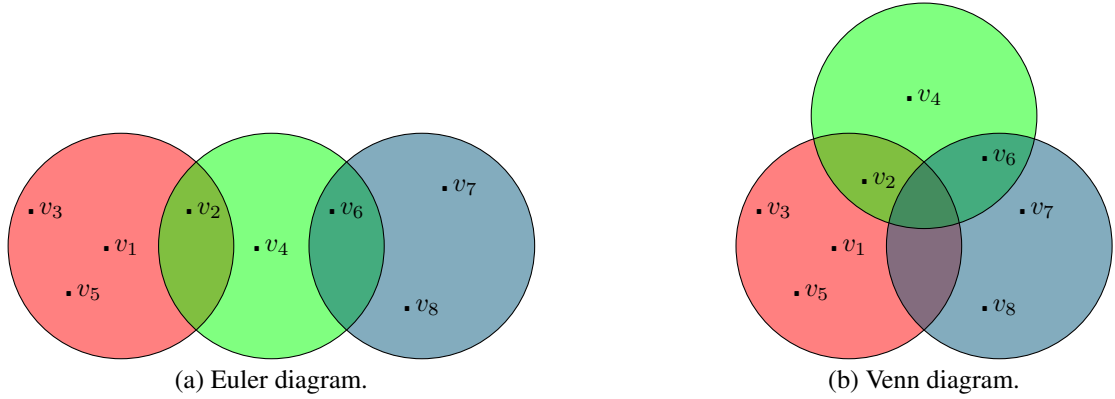
\begin{figure}
\begin{center}

\begin{tabular}{>{\centering}p{0.5\textwidth}>{\centering}p{0.5\textwidth}}
\def\firstcircle{(0,0) circle (1.5cm)} 
\def\secondcircle{(0:2cm) circle (1.5cm)} 
\def\thirdcircle{(0:4cm) circle (1.5cm)}
\begin{tikzpicture}     

\begin{scope}[fill opacity=0.5]         
\fill[red] \firstcircle;         
\fill[green] \secondcircle;         
\fill[blue] \thirdcircle;         
\draw \firstcircle;         
\draw \secondcircle;         
\draw \thirdcircle;
\end{scope}
\node at (0,0) {$\centerdot\,v_1$};
\node at (-1,0.5) {$\centerdot\,v_3$};
\node at (-0.5,-0.6) {$\centerdot\,v_5$};
\node at (1.1,0.5) {$\centerdot\,v_2$};
\node at (2,0) {$\centerdot\,v_4$};
\node at (3,0.5) {$\centerdot\,v_6$};
\node at (4.5,.8) {$\centerdot\,v_7$};
\node at (4,-0.8) {$\centerdot\,v_8$};

\end{tikzpicture} & \def\firstcircle{(0,0) circle (1.5cm)} 
\def\secondcircle{(60:2cm) circle (1.5cm)} 
\def\thirdcircle{(0:2cm) circle (1.5cm)}
\begin{tikzpicture}     

\begin{scope}[fill opacity=0.5]         
\fill[red] \firstcircle;         
\fill[green] \secondcircle;         
\fill[blue] \thirdcircle;         
\draw \firstcircle;         
\draw \secondcircle;         
\draw \thirdcircle;
\end{scope}
\node at (0,0) {$\centerdot\,v_1$};
\node at (-1,0.5) {$\centerdot\,v_3$};
\node at (-0.5,-0.6) {$\centerdot\,v_5$};
\node at (0.4,0.9) {$\centerdot\,v_2$};
\node at (1,2) {$\centerdot\,v_4$};
\node at (2,1.2) {$\centerdot\,v_6$};
\node at (2.5,0.5) {$\centerdot\,v_7$};
\node at (2,-0.8) {$\centerdot\,v_8$};

\end{tikzpicture} \tabularnewline
(a) Euler diagram. & (b) Venn diagram.\tabularnewline
\end{tabular}

\end{center}

\caption{Difference between (a) Euler diagram and (b) Venn diagram.}
\end{figure}

In \citet{makinen1990how}, the author presents the Venn diagram representation
of hypergraphs: a Venn representation of a hypergraph represents hyperedges
as closed curves that contains points symbolizing the vertices. In
\citet{johnson23hypergraph}, the authors present two representations
with the subset standard. The first representation is a hyperedge-based
Venn diagram based on the representation of the hypergraph as a planar
graph, an embedding of this planar graph and a one-to-one map between
the hyperedges and the embedding faces, such that taking any subset
of the vertices, the union of the faces corresponding to this subset
includes a region of the plane with a connected interior. Hypergraphs
that can be represented in such a way are said hyperedge-planar. The
second representation is a vertex-based Venn diagram where the mapping
is from the set of vertices to the embedding faces such that for every
hyperedge the interior of the face union that contains all the hyperedge
vertices is connected, with a corresponding notion of vertex-planar
when such a representation is feasible.

In \citet{bertault2001drawing}, a drawing system is presented, named
PATATE, where hypergraphs are represented using the subset standard:
a graph is built combining three possibilities. The first representation
is achieved by building the incidence graph of the hypergraph, adding
a dummy vertex for each hyperedge that is linked by an edge, placed
at the barycenter. The second representation is built using a minimal
Euclidean spanning-tree that covers the vertices of each hyperedge:
the edges of the spanning-tree are added to the graph. The third representation
is built using a minimal Steiner tree: a vertex is added for every
Steiner point and the edges of the tree are added. A force directed
algorithm is applied to the built graph and the convex hull of each
hyperedge is then computed.

The Venn diagram representation is also addressed in \citet{estrada2005complex}.
The Venn diagram is relevant for small hypergraphs but is harder to
use for large hypergraphs.

A large survey on set representations is available in \citet{alsallakh2016stateoftheart}.

\subsubsection{The edge standard}

The edge standard includes two main representations: the \textbf{clique
expansion} \index{hypergraph!representation!clique expansion}\index{clique expansion!hypergraph}\index{representation!hypergraph!clique expansion}and
the \textbf{extra-node representation}\index{hypergraph!representation!extra node}\index{extra node!hypergraph}\index{representation!hypergraph!extra node}
also called the incidence representation in \citet{kaufmann2009subdivision}.

The clique expansion is based on the 2-section of the hypergraph where
a hyperedge of cardinality $n$ is represented by $\dfrac{n\left(n-1\right)}{2}$
undirected edges.

The extra vertex representation corresponds to the incidence graph
of the hypergraph, which is a bipartite graph. It is called the König
representation of the hypergraph in \citet{zykov1974hypergraphs}.
In this case, a hyperedge of size $n$ is represented with only $n$
edges connecting hyperedge vertices to a central vertex called \textbf{extra-node}.

Figure \ref{Fig : Clique view vs Extra node view of an Hyperedge}
illustrates these two representations. 

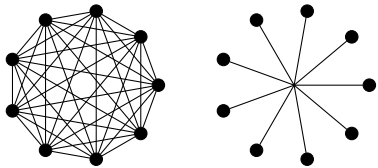
\begin{figure}[H]
\begin{center}%
\begin{tabular}{cc}
\newcount\tempcount
\def \n {9}
\def \mn {8}
\def \radius {1cm}
\begin{tikzpicture}[transform shape, every node/.style = {scale = 0.5}]
\foreach \s in {1,...,\n} {   
     \node[draw, circle, fill=black] (N-\s) at ({360/\n * (\s - 1)}:\radius) {};
}    
\foreach \ni in {1,...,\mn}{
   \tempcount=\ni
   \advance\tempcount by 1
   \foreach \nj in {\the\tempcount,...,\n}{
      \path (N-\ni) edge (N-\nj);
   }
}
\end{tikzpicture} & \newcount\tempcount
\def \n {9}
\def \mn {8}
\def \radius {1cm}
\begin{tikzpicture}[transform shape, every node/.style = {scale = 0.5}]
\node[draw,circle,fill=white,scale=0.1] (N-0) at (0,0) {};
\foreach \s in {1,...,\n} {   
     \node[draw, circle, fill=black] (N-\s) at ({360/\n * (\s - 1)}:\radius) {};
}    
\foreach \ni in {1,...,\n}{
   \path (N-0) edge (N-\ni);
}
\end{tikzpicture}\tabularnewline
\end{tabular}\end{center}

\caption{Clique expansion vs extra-node representation of a hyperedge.}
\label{Fig : Clique view vs Extra node view of an Hyperedge}
\end{figure}

The clique expansion generates graph with more edges than the extra-node
representation for values of $n\geqslant4$. Moreover, if a hyperedge
$e_{j_{1}}$ is a strict subset of an other hyperedge $e_{j_{2}},$
the clique expansion of $e_{j_{1}}$ does not add any other edge,
while the extra-node representation brings $\left|e_{j_{1}}\right|$
new edges and an additional vertex. An example of an unfavorable case
is given in Figure \ref{Fig : Extra node unfavorable cases}.

\begin{figure}[H]
\resizebox{\columnwidth}{!}{%
\begin{tabular}{|c|c|}
\hline 
Clique view & extra-node view\tabularnewline
\hline 
\newcount\tempcount
\def \n {5}
\def \mn {4}
\def \radius {1cm}
\begin{tikzpicture}[transform shape, every node/.style = {scale = 0.5}]
\foreach \s in {1,...,\n} {   
     \node[draw, circle, fill=black] (N-\s) at ({360/\n * (\s - 1)}:\radius) {};
}    
\foreach \ni in {1,...,\mn}{
   \tempcount=\ni
   \advance\tempcount by 1
   \foreach \nj in {\the\tempcount,...,\n}{
      \path (N-\ni) edge (N-\nj);
   }
}
\end{tikzpicture} & \newcount\tempcount
\def \n {5}
\def \mn {4}
\def \radius {1cm}
\begin{tikzpicture}[transform shape, every node/.style = {scale = 0.5}]
\node[draw,circle,fill=white,scale=0.1,color=red] (N-0) at (0,0) {};
\foreach \s in {1,...,\n} {   
     \node[draw, circle, fill=black] (N-\s) at ({360/\n * (\s - 1)}:\radius) {};
}    
\foreach \ni in {1,...,\n}{
   \path[color=red] (N-0) edge (N-\ni);
}
\node[draw,circle,fill=white,scale=0.1,color=green] (N-6) at (0.5,0.25) {};
\foreach \ni in {1,...,3}{
   \path[color=green] (N-6) edge (N-\ni);
}
\node[draw,circle,fill=white,scale=0.1,color=blue] (N-7) at (-0.5,-0.25) {};
\foreach \ni in {3,...,5}{
   \path[color=blue] (N-7) edge (N-\ni);
}
\end{tikzpicture}\tabularnewline
\hline 
In this case: 10 edges, 5 vertices & In this case: 11 edges, 5 vertices and 3 extra-nodes\tabularnewline
\hline 
\end{tabular}

}

\caption{Unfavorable case for the extra-node view.}
\label{Fig : Extra node unfavorable cases}
\end{figure}

In the clique representation, vertices of hyperedges are seen as interacting
with 2-adic relationships and the information on the meso-structure
is lost \citet{taramasco2010academic}. The extra-node representation
allows to keep the multi-adic relationships. Moreover, we have shown
in \citet{ouvrard2017networks} that the number of edges can be substantially
reduced by the extra-node representation for a power-law hypergraph.

\citet{junghans2008visualization} focuses on hyperedge drawing so
that they do not intersect cluster groups, giving a solution based
on force attraction/repulsion.

Other representations in the node standard exist: in \citet{paquette2011hypergraph},
the authors propose a representation using a pie-chart node approach,
which is valuable for hypergraphs where the hyperedges do not intersect
one another with too high cardinality. In \citet{kerren2013novel},
a radial approach is presented that fits for small hypergraphs. In
\citet{kaufmann2009subdivision}, a Steiner tree representation of
the hyperedges is considered as the edge representation of a hypergraph.

\subsubsection{The Zykov representation}

In \citet{zykov1974hypergraphs}, the author tackles the representation
of hypergraphs and proposes a notion of hypergraph planarity. A planar
hypergraph is a hypergraph where the vertices are represented as points
in a plane and hyperedges as closed curves not intersecting one another
but in the neighborhood of a common incident vertex of hyperedges. 

This notion of Zykov-planarity corresponds to the planarity of the
incidence graph of the hypergraph.

The author proposes a representation, called now the Zykov representation\index{hypergraph!representation!Zykov}\index{representation!hypergraph!Zykov}\index{Zykov representation}
and shown in Figure \ref{Fig: Zykov deformation} (c), that is based
on a continuous deformation of the edges, starting with the Euler
diagram in Figure \ref{Fig: Zykov deformation} (a) via the PaintSplash
representation in Figure \ref{Fig: Zykov deformation} (b) and that
stops just before obtaining the incidence graph as shown in Figure
\ref{Fig: Zykov deformation} (d). Hyperedges are represented as faces
of a subdivision realized by the vertices belonging to the corresponding
hyperedge. Two hyperedges with a common vertex are incident at that
vertex. The representation in the original article is in black and
white; hence, this representation requires two colors: one for the
background and an other one for the faces---we keep here the colors
for helping to the understanding of the continuous deformation of
the hyperedges from the Euler representation to the extra-node representation.
This representation is intensively used with simplicial complexes
which are particular case of hypergraphs.

The Zykov representation cannot be totally put neither in the edge
standard nor in the subset standard.

\begin{figure}
\begin{center}\resizebox{\textwidth}{!}{
\begin{tikzpicture} 
\node[inner sep=0pt] (fig1) at (0,0){
	\begin{tabular}{c}
\rule{0pt}{150pt}
         \includegraphics[width=.45\textwidth]{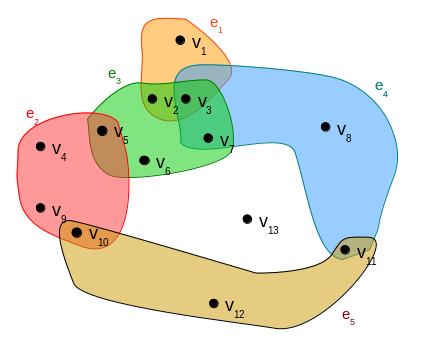}\\
		(a) Euler diagram.
	\end{tabular}
}; 
\node[inner sep=0pt] (fig2) at (8,0)     {
	\begin{tabular}{c}
\rule{0pt}{150pt}
\includegraphics[width=.45\textwidth]{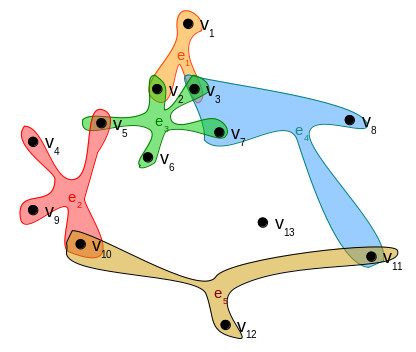}\\
(b) PaintSplash representation.
\end{tabular}
}; 
\node[inner sep=0pt] (fig3) at (0,-7)     {
	\begin{tabular}{c}
\rule{0pt}{150pt}
\includegraphics[width=.45\textwidth]{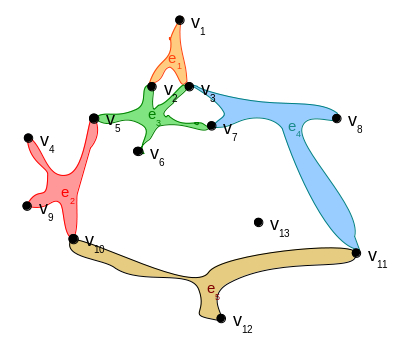}\\
(c) Zykov representation.
	\end{tabular}
}; 
\node[inner sep=0pt] (fig4) at (8,-7)     {
	\begin{tabular}{c}
\rule{0pt}{150pt}	\includegraphics[width=.45\textwidth]{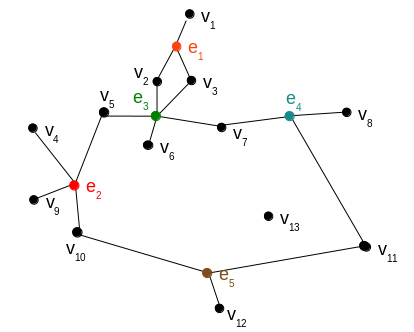}\\
(d) Extra-node representation.
	\end{tabular}
}; 
\draw[->,thick] (fig1.east) -- (fig2.west);
\draw[->,thick] (fig2.south west) -- (fig3.north east);
\draw[->,thick] (fig3.east) -- (fig4.west);
\end{tikzpicture}
}\end{center}

\caption{Moving from the Euler diagram (a) via the PaintSplash representation
(b) and the Zykov representation (c) to the incident representation
(d) based on \citet{zykov1974hypergraphs}---in the order of the
arrows. The original figure is in black and white.}

\label{Fig: Zykov deformation}
\end{figure}
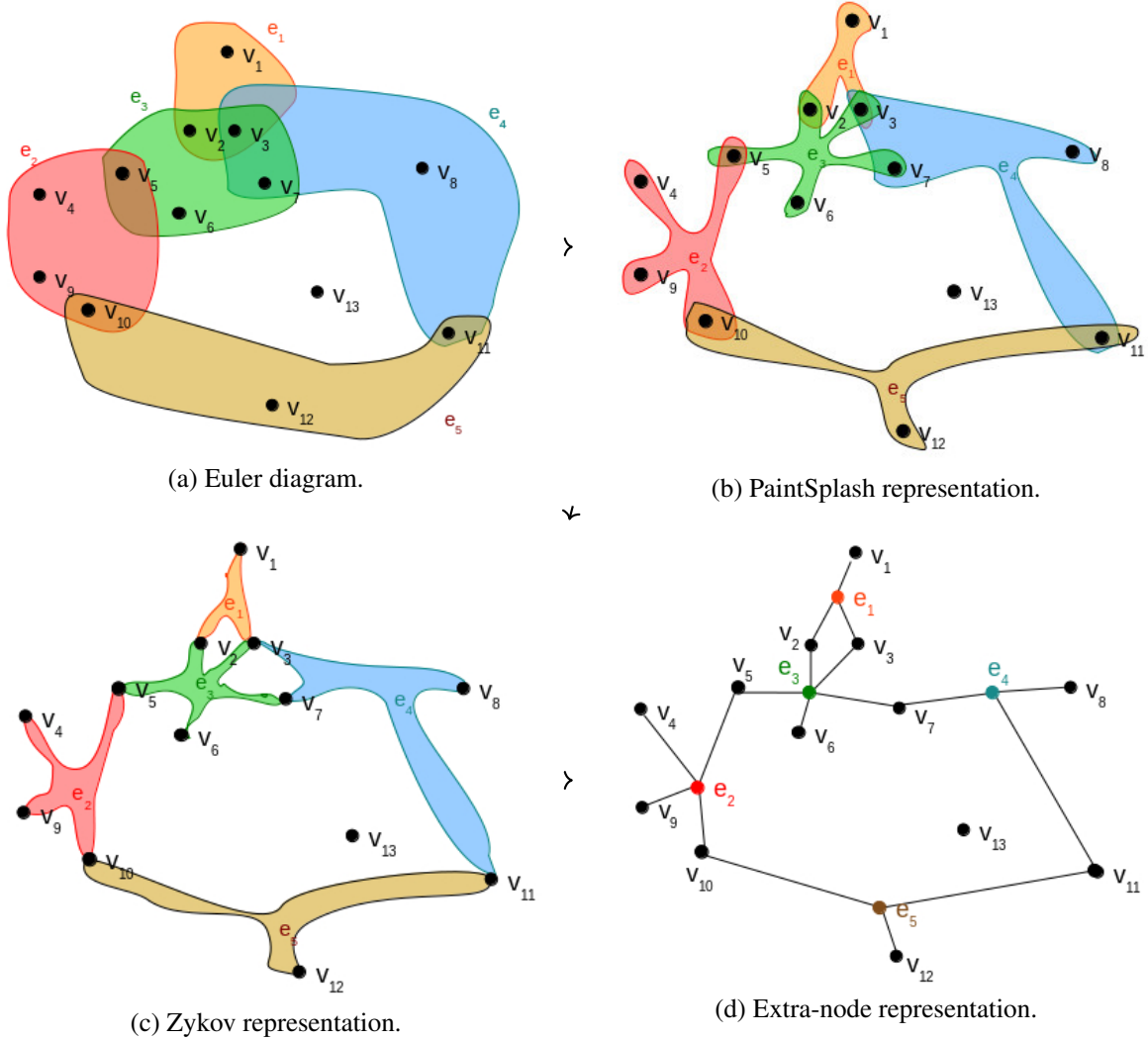

\subsubsection{PaintSplash representation}

It is worth mentioning that Figure \ref{Fig: Zykov deformation} (b)
can lead to a good compromise between the subset standard and the
edge standard, that we have called in its modern colored version the
PaintSplash representation of a hypergraph\index{hypergraph!representation!PaintSplash}\index{representation!hypergraph!PaintSplash}\index{PaintSplash representation},
presented in Figure \ref{Fig: PaintSplash}.

\begin{figure}
\begin{center}\includegraphics[scale=0.35]{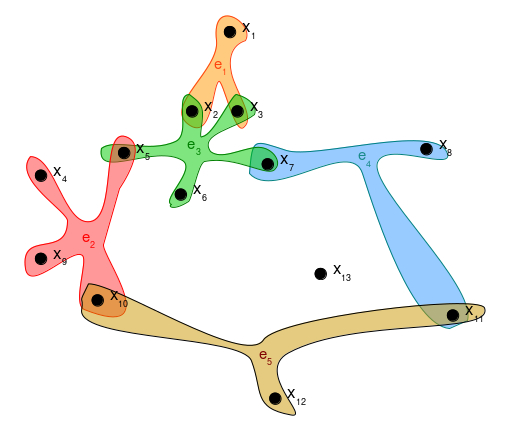}\end{center}

\caption{PaintSplash representation of a hypergraph.}

\label{Fig: PaintSplash}
\end{figure}

\subsubsection{Large hypergraph visualisation}

\begin{figure}[p]
\begin{center}\includegraphics[scale=0.4]{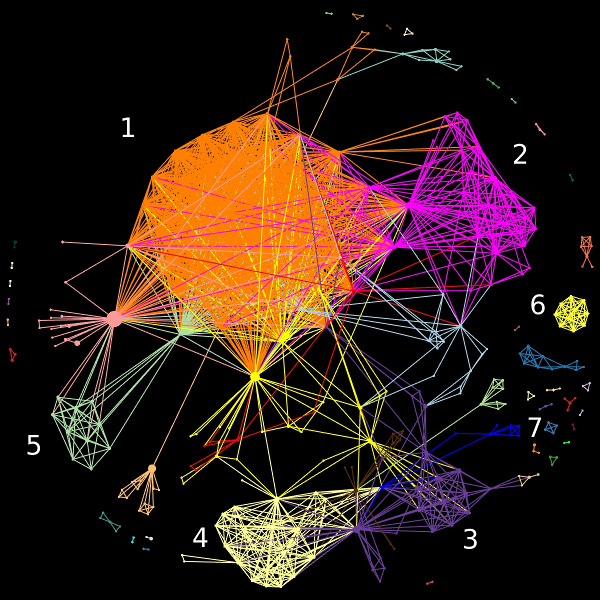}\end{center}

Sub-figure \ref{Fig exp : title:((bgo AND cryst*) OR (bgo AND calor*)) abstract:((bgo AND cryst*) OR (bgo AND calor*)) : organisations eng}
(a): Clique representation: The coordinates of the nodes are calculated
by ForceAtlas2 on the extra-node view and then transferred to this
view.\newline

\begin{center}\includegraphics[scale=0.8]{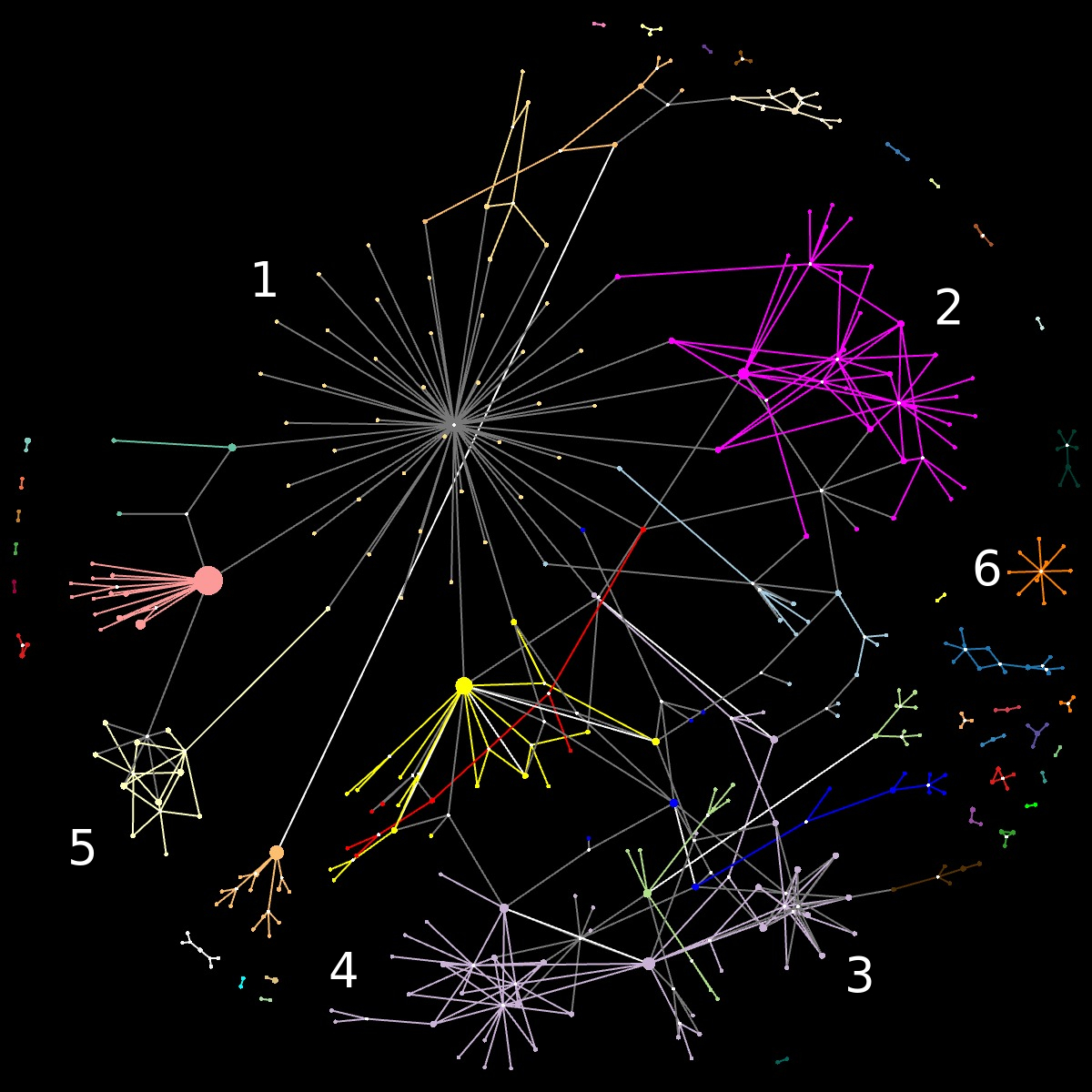}\end{center}

Sub-figure \ref{Fig exp : title:((bgo AND cryst*) OR (bgo AND calor*)) abstract:((bgo AND cryst*) OR (bgo AND calor*)) : organisations eng}
(b): Extra-node representation: The coordinates of the nodes are calculated
by ForceAtlas2 for this representation.

\caption{Hypergraph of organizations: Sub-figures (a) and (b) refer to the
search:\protect \\
title:((bgo AND cryst{*}) OR (bgo AND calor{*})) abstract:((bgo AND
cryst{*}) OR (bgo AND calor{*})) from \citet{ouvrard2017networks}.}
\label{Fig exp : title:((bgo AND cryst*) OR (bgo AND calor*)) abstract:((bgo AND cryst*) OR (bgo AND calor*)) : organisations eng}
\end{figure}

\begin{figure}[p]
\begin{center}\includegraphics[scale=1.75]{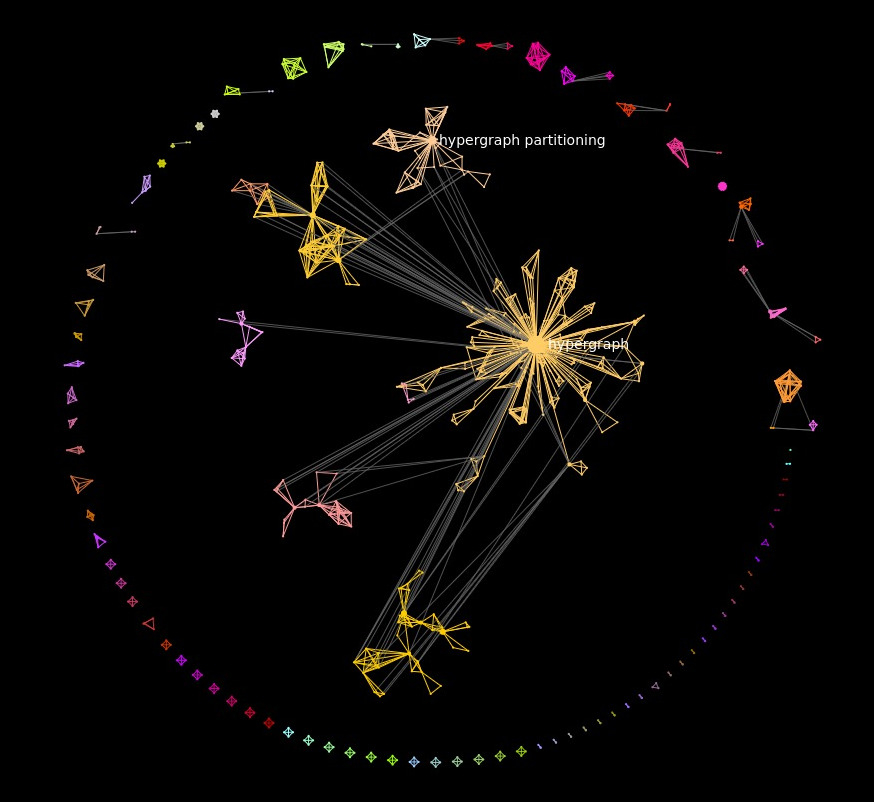}\end{center}

{\scriptsize{}\caption{Keyword collaborations from search: ``TITLE: hypergraph''.}
}{\scriptsize\par}

\begin{center}\#publications = 200, \#nodes = 707, \#edges = 1655,
\#clusters = 102, \#isolated clusters = 67.\end{center}\label{Fig: Large hypergraph example 1}
\end{figure}

\begin{figure}[p]
\begin{center}\includegraphics[scale=1.75]{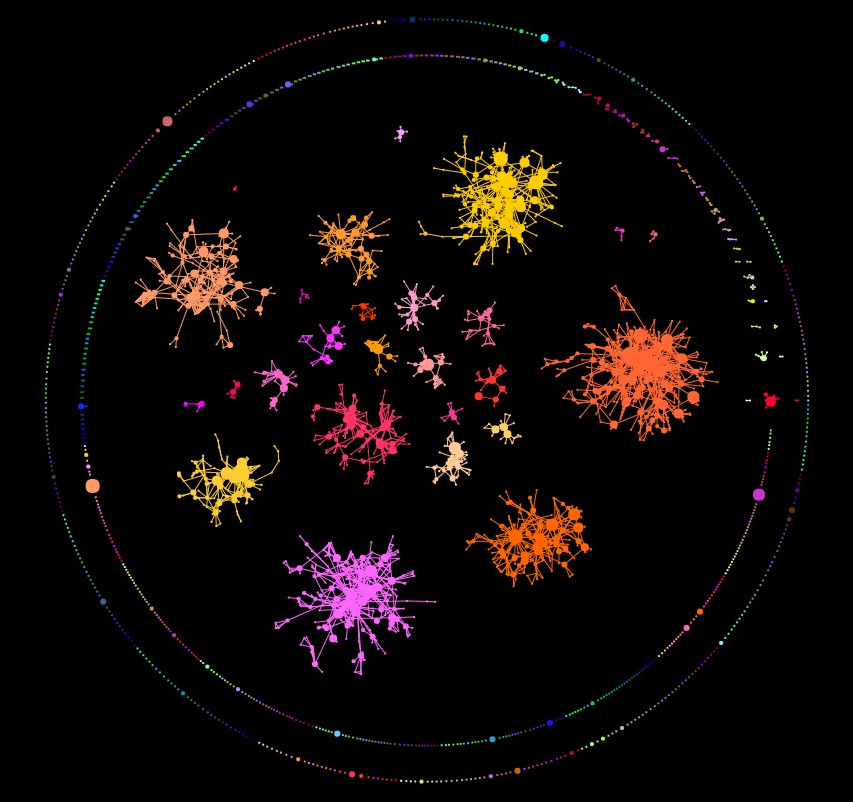}\end{center}

\caption{Organization collaborations from search: ``TITLE:graph''.}

\begin{center}\#publications = 3969, \#patents=893

\#nodes = 2932, \#edges = 4731, \#clusters = 951, \#isolated clusters
= 914.\end{center}\label{Fig: Large hypergraph example 2}
\end{figure}

\begin{figure}
\begin{center}\hspace*{0.2cm}
\begin{tikzpicture}[->,>=stealth',scale=0.53, every node/.append style={transform shape}]


\node[state=blue,
		yshift=-1cm,
		minimum width=25cm,
		minimum height=16.5cm,
		anchor=center] (H) {
		};
\node[] (H_Title) at ([yshift=-1.5em,xshift=-11cm]H.north) {\textbf{Hypergraph}};
\node[state=green!50!black,minimum width=4cm] (H_CC) at ([yshift=-3em,xshift=1cm]H_Title.north) {\begin{tabular}{c}
				Connected components
			\end{tabular}};
\node[state=green!50!black,minimum width=4cm] (H_CCQH) at ([xshift=3cm]H_CC.east) {
			\begin{tabular}{c}
				Connected components\\
				Quotient Hypergraph
			\end{tabular}
		};
\node[state=green!50!black,minimum width=3cm] (H_CCQHL) at ([xshift=5cm]H_CCQH.east) {
			\begin{tabular}{c}
				Connected components\\
				Quotient Hypergraph\\
				Layout
			\end{tabular}
};
\node[state=red,minimum width=3cm] (H_CCL) at ([xshift=-2cm,yshift=-7cm]H.east) {
			\begin{tabular}{c}
				Hypergraph\\
				layout
			\end{tabular}
	};
\node[state=blue,minimum width=3cm] (H_join) at ([xshift=-2cm,yshift=-5cm]H.east) {+};

\node[state=green!50!black,
		right of=H,
		yshift=-0.5cm,
		node distance=-1.5cm,
		minimum width=19cm,
		minimum height=12.5cm,
		anchor=center] (CC) {};
\node[] (CC_Title) at ([yshift=-1.5em,xshift=-7cm]CC.north) {\textbf{Connected component}};
\node[state=orange!70!blue,minimum width=4cm] (CC_D) at ([yshift=-3em,xshift=0cm]CC_Title.north) {Communities detection};
\node[state=orange!70!blue,minimum width=4cm] (CC_QH) at ([xshift=3cm]CC_D.east) {
			\begin{tabular}{c}
				Communities\\
				Quotient Hypergraph
			\end{tabular}
	};
\node[state=orange!70!blue,minimum width=3cm] (CC_CCL) at ([xshift=3.3cm]CC_QH.east) {
			\begin{tabular}{c}
				Communities\\
				Quotient\\
				Hypergraph\\
				Layout
			\end{tabular}
	};
\node[state=orange!70!blue,minimum width=3cm] (CC_SC) at ([xshift=2.3cm]CC_CCL.east) {
			\begin{tabular}{c}
				Size and center\\
				of connected\\
				component
			\end{tabular}
	};


\node[state=orange!70!blue,
		right of=CC,
		yshift=-1cm,
		node distance=-1.5cm,
		minimum width=15cm,
		minimum height=7cm,
		anchor=center] (C) {};
\node[] (C_Title) at ([yshift=-1.5em,xshift=-6cm]C.north) {\textbf{Community}};
\node[state=orange,minimum width=3cm] (C_Coa) at ([yshift=-3cm,xshift=1cm]C_Title.north) {Coarsening};

\node[state=orange,
		right of=C_Coa,
		yshift=0cm,
		node distance=4.5cm,
		minimum width=4cm,
		minimum height=6cm,
		anchor=center] (C_L) {};
\node[] (C_L_Title) at ([yshift=-1.5em,xshift=0cm]C_L.north) {Layout};
\node[state=purple,minimum width=3cm] (C_L_IN) at ([yshift=-3em,xshift=0cm]C_L_Title.north) {Important nodes};
\node[state=purple,minimum width=0.5cm] (C_L_PIN) at ([yshift=-0.5cm,xshift=0cm]C_L_IN.south) {+};
\node[state=purple,minimum width=3cm] (C_L_ON) at ([yshift=-0.5cm,xshift=0cm]C_L_PIN.south) {Other nodes};
\node[state=purple,minimum width=0.5cm] (C_L_PON) at ([yshift=-0.5cm,xshift=0cm]C_L_ON.south) {+};
\node[state=purple,minimum width=3cm] (C_L_Ov) at ([yshift=-0.65cm,xshift=0cm]C_L_PON.south) {
			\begin{tabular}{c}
				Overlapping\\
				optimisation
			\end{tabular}
		};

\node[state=orange,minimum width=3cm] (C_SC) at ([yshift=0cm,xshift=3.5cm]C_L.east) {
			\begin{tabular}{c}
				Size and center\\
				of community
			\end{tabular}};


\fill[fill=green!50!black] ([xshift=2cm]H_CCQH.east)--([xshift=-1cm]CC.north)--([xshift=1cm]CC.north);
\fill[fill=orange!70!blue] ([xshift=1cm]CC_QH.east)--([xshift=-1cm]C.north)--([xshift=1cm]C.north);

\draw[green!50!black,line width=0.25mm] (H_CC.east) -> (H_CCQH.west);
\draw[green!50!black,line width=0.25mm] (H_CCQH.east) -> (H_CCQHL.west);
\draw[orange!70!blue,line width=0.25mm] (CC_D.east) -> (CC_QH.west);
\draw[orange!70!blue,line width=0.25mm] (CC_QH.east) -> (CC_CCL.west);
\draw[orange!70!blue,line width=0.25mm] (CC_CCL.east) -> (CC_SC.west);

\draw[orange,line width=0.25mm] (C_Coa.east) -> (C_L.west);
\draw[orange,line width=0.25mm] (C_L.east) -> (C_SC.west);
\draw[red,line width=0.25mm] (H_join.south) -> (H_CCL.north);

\draw[green!50!black,line width=0.25mm,post,rounded corners=5pt] (H_CCQHL.east)-|([xshift=0.5cm]H_join.north) node[] at ([xshift=3cm,yshift=1em]H_CCQHL.east) {$\left(x_{CC}-\omega_{CC},y_{CC}-\omega_{CC}\right)$};

\draw[orange!70!blue,line width=0.25mm,post,rounded corners=5pt] (CC_SC.east)-|(H_join.north) node[] at ([xshift=2cm,yshift=1em]CC_SC.east) {$\left(x_{Cl}-\omega_{Cl},y_{Cl}-\omega_{Cl}\right)$};

\draw[orange,line width=0.25mm,post,rounded corners=5pt] (C.east)-|([xshift=-0.5cm]H_join.north) node[] at ([xshift=2cm,yshift=1em]C.east) {$\left(x_N-\omega_C,y_N-\omega_C\right)$};

\end{tikzpicture}\end{center}

\caption{Principle of the calculation of coordinates for large hypergraphs.}
\label{Fig: Vertex layout}
\end{figure}
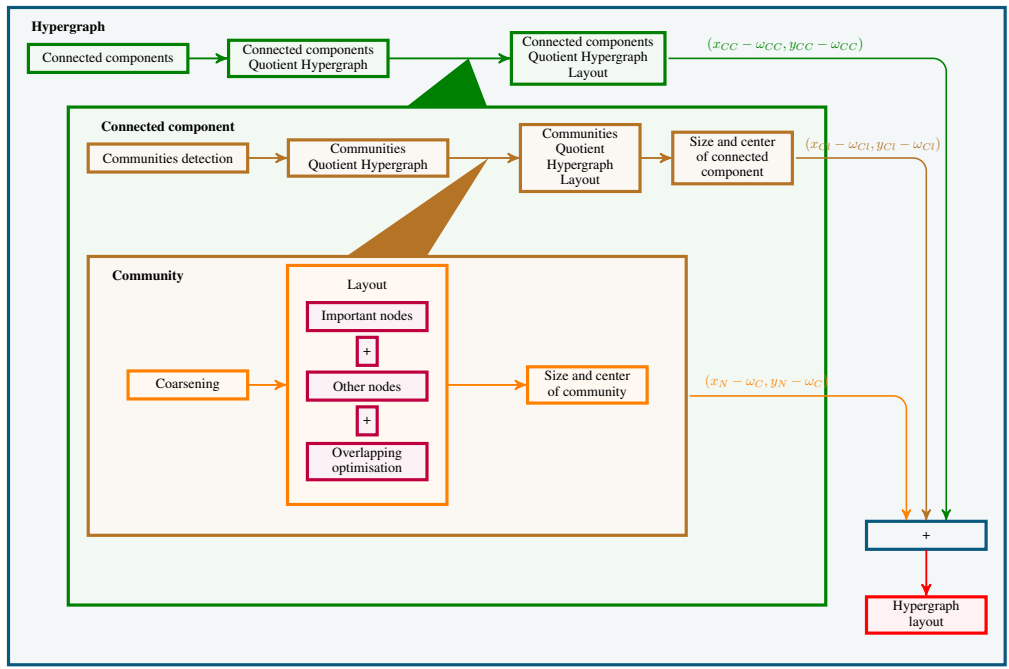

Large hypergraphs require strategies to be properly visualized. We
have shown in \citet{ouvrard2017networks} that switching from the
2-section representation to the incidence representation reduces considerably
the overall cognitive load, as the number of edges to be represented
is effectively highly diminished in real networks; moreover, true
collaborations are seen more distinguishably. An example is given
in Figure \ref{Fig exp : title:((bgo AND cryst*) OR (bgo AND calor*)) abstract:((bgo AND cryst*) OR (bgo AND calor*)) : organisations eng},
that shows how the representation in the two modes impact on the cognitive
load.

Other strategies have been implemented in order to improve the global
layout and information retrieval of large hypergraphs. It includes
the segmentation of the hypergraph into connected components, ordering
them by the number of references they represent, and then with respect
to the number of vertices within co-occurrences they contain. The
layout of each connected component is computed separately, by considering
the communities based on the Louvain algorithm presented in \citet{blondel2008fast}
and using a force directed algorithm, named ForceAtlas2, presented
in \citet{jacomy2014forceatlas2}, to ensure the layout of the communities
and of the vertices inside the communities. To achieve a proper layout,
vertices of importance for the community are placed first and the
remaining vertices are then placed around those fixed important vertices,
with an optimization of the layout. The principle of the calculus
is given in Figure \ref{Fig: Vertex layout} and the results is shown
in Figure \ref{Fig: Large hypergraph example 1} and Figure \ref{Fig: Large hypergraph example 2}.

Optimizing the layout of a large hypergraph is a challenging task.
It requires to capture important vertices of the hypergraph, i.e.\
the ones we do not want to loose by overwhelming or overlapping in
the representation. This requirement has been the starting point of
our quest for a diffusion tensor.

\section{Morphisms, category and hypergraphs}

\label{subsec:Morphisms,-category-and hypergraphs}

\subsubsection{A parenthesis on categories}

The notions presented in this paragraph are taken from different sources
including the page on Category\footnote{\href{https://en.wikipedia.org/wiki/Category_(mathematics)}{https://en.wikipedia.org/wiki/Category\_(mathematics)}}
and the reference book \citet{adamek2004abstract}.

A \textbf{category}\index{category} $\mathcal{C}$ is formed by:
\begin{itemize}
\item a class of \textbf{objects}\index{object}\index{category!object}
denoted $\text{Ob}\left(\mathcal{C}\right)$
\item a class of elements called \textbf{arrows}\index{arrow}\index{category!arrow}
or \textbf{morphisms}\index{morphism}\index{category!morphism} $\text{mor}\left(\mathcal{C}\right)$
\end{itemize}
and such that:
\begin{itemize}
\item given a pair $\left(X,Y\right),$ with $X\in\text{Ob}\left(\mathcal{C}\right)$
and $Y\in\text{Ob}\left(\mathcal{C}\right),$ there exists a set:
$\mbox{\text{Hom}\ensuremath{\left(X,Y\right)\in\text{mor}\left(\mathcal{C}\right)}}$
called the morphisms from $X$ to $Y$ in $\mathcal{C}.$ If $f$
is such an homomorphism, we write it $f:X\rightarrow Y.$ $X$ is
called the source and $Y$ the target.
\item for every $X\in\text{Ob}\left(\mathcal{C}\right),$ there exists a
morphism written $\text{id}_{X}\in\text{Hom}\left(X,X\right)$ called
the identity on $X.$
\item a composition law $\circ$ such that for two morphisms $f:A\rightarrow B$
and $g:B\rightarrow C$ are associated to a morphism $g\circ f:A\rightarrow C$
called composition of $f$ and $g$, such that:
\item the composition is an associative law: for any $f\in\text{Hom}\left(A,B\right)$,
$g\mbox{\ensuremath{\in\text{Hom}\left(B,C\right)}}$ and $h\in\text{Hom}\left(C,D\right):$
\[
\left(h\circ g\right)\circ f=h\circ\left(g\circ f\right)
\]
\item the identity are the neutral elements of the composition: given any
morphism $f\in\text{Hom}\left(A,B\right):$ 
\[
f\circ\text{Id}_{A}=f=\text{Id}_{B}\circ f
\]
\end{itemize}
For instance, the category where the objects are sets and where morphisms
are applications between sets with the natural composition of applications
is the Set category, written $\text{Set}$.

As summarized in \citet{isah2015concept}, different kinds of morphisms
exist. We consider a morphism $f:A\rightarrow B$ of a category $\mathcal{C}$
where $A$ and $B$ are two objects of $\mathcal{C}.$ $f$ is:
\begin{itemize}
\item a \textbf{monomorphism}\index{category!monomorphism}\index{monomorphism}
if for all pairs of morphisms of $\mathcal{C}$ with $h:C\rightarrow A$
and $k:C\rightarrow A,$ $f\circ h=f\circ k$ implies $h=k$ i.e.\
if $f$ is left-cancellable.
\item a \textbf{split monomorphism} (or \textbf{section})\index{category!split monomorphism}\index{split monomorphism}
if it is left-invertible i.e.\ there exists a morphism $g:B\rightarrow A$
of $\mathcal{C}$ such that: $g\circ f=\text{Id}_{A}.$
\item an \textbf{epimorphism}\index{category!epimorphism}\index{epimorphism}
if it is right-cancellable i.e.\ for all pairs $h:B\rightarrow C$
and $k:B\rightarrow C$, $h\circ f=k\circ f$ implies $h=k.$
\item a \textbf{split epimorphism} (or \textbf{retraction})\index{category!epimorphism}\index{epimorphism}
if it is right-invertible i.e.\ there exists $g:B\rightarrow A$
such that $f\circ g=\text{Id}_{B}.$
\item a \textbf{bimorphism}\index{category!bimorphism}\index{bimorphism}
if it is both a monomorphism and an epimorphism.
\item an \textbf{isomorphism}\index{category!isomorphism}\index{isomorphism}
if it is both a split monomorphism and a split epimorphism.
\end{itemize}
A \textbf{sub-category}\index{category!sub-category}\index{sub-category}
$\mathcal{S}$ of a category $\mathcal{C}$ is a category where the
objects are objects of $\mathcal{C}$ and where the morphisms are
morphism of $\mathcal{C}$ between objects of $\mathcal{S}.$

When a sub-category $\mathcal{S}$ of a category $\mathcal{C}$ has
its morphisms that are exactly all the morphisms of $\mathcal{C}$
between objects of $\mathcal{S},$ this sub-category is said \textbf{full}\index{category!full}\index{full category}.

Correspondences between two categories are possible by associating
to each object (resp. a morphism) of the first category, an object
(resp. a morphism) of the second category: such correspondences are
called functors.

Let $\mathcal{C}$ and $\mathcal{D}$ be two categories. A \textbf{functor}\index{functor}\index{category!functor}
from $\mathcal{C}$ to $\mathcal{D}$ (or covariant functor), written
$F:\mathcal{C}\rightarrow\mathcal{D},$ is defined by giving:
\begin{itemize}
\item a function such that for all $X\in\text{Obj}\left(\mathcal{C}\right),$
there exists an object $F\left(X\right)\in\text{Obj}\left(\mathcal{D}\right);$
\item a function such that for all $f\in\text{mor}\left(\mathcal{C}\right)$,
with $f:X\rightarrow Y,$ there exists $F\left(f\right)\in\text{mor}\left(\mathcal{D}\right)$
such that: $F(f):F\left(X\right)\rightarrow F\left(Y\right);$
\end{itemize}
and, such that:
\begin{itemize}
\item for all $X\in\mathcal{C},$ we have (identity conservation): 
\[
F\left(\text{id}_{X}\right)=\text{id}_{F(X)};
\]
\item for all $X,Y,Z\in\mathcal{C}$ and $f\in\text{Hom}\left(X,Y\right)$
and $g\in\text{Hom}\left(Y,Z\right),$ it holds (composition conservation):
\[
F\left(g\circ f\right)=F\left(g\right)\circ F\left(f\right).
\]
\end{itemize}
The construction of new categories can be achieved by using the \textbf{dual
of a category} $\mathcal{C},$\index{category!dual}\index{dual!category}
written $\mathcal{C}^{\text{op}},$ where the objects are the ones
of $\mathcal{C}$ and the arrows of $\mathcal{C}^{\text{op}}$ are
the ones of $\mathcal{C}$ taken reversely. An other usual way is
to consider the product category of two categories $\mathcal{C}$
and $\mathcal{C}^{\prime},$ denoted $\mathcal{C}\times\mathcal{C}^{\prime},$
where the objects are pairs of objects of $\mathcal{C}$ and $\mathcal{C}^{\prime}$
taken in this order and the arrows are constituted of pairs of arrows
of $\mathcal{C}$ and $\mathcal{C}^{\prime}$ taken in this order.

\subsubsection{Hypergraph homomorphism and category}

In \citet{dorfler1980categorytheoretical}, the authors use the category
theory to consider hypergraph product. In this article, a hypergraph
$\left(V,E,f\right)$ is defined as a triple with $V$ the vertex
set, $E$ the hyperedge set and $f:E\rightarrow\mathcal{P}\left(V\right)\backslash\left\{ \emptyset\right\} $
a function that assigns to each hyperedge its non-empty set of vertices.
Duplicated hyperedges are allowed, i.e.\ two hyperedges can have
the same image by $f.$

The following definitions and properties are taken from \citet{dorfler1980categorytheoretical}.

The \textbf{covariant power set functor}\index{category!covariant power set functor}\index{covariant power set functor}
$\mathcal{P}:\text{Set}\rightarrow\text{Set}$ is defined by:
\begin{itemize}
\item for $A\in\text{Obj}\left(\text{Set}\right)$, $\mathcal{P}\left(A\right)$
is the set of all subsets of $A;$
\item for $h\in\text{mor\ensuremath{\left(\text{Set}\right)}}$, $h:A\rightarrow B$,
$\mathcal{P}\left(h\right)$ is the homomorphism: $\mathcal{P}\ensuremath{\left(A\right)}\rightarrow\mathcal{P}\left(B\right)$
such that for all $A'\subset A$, $\mathcal{P}\left(h\right):A'\rightarrow h\left(A'\right).$
\end{itemize}
Considering two hypergraphs $\mathcal{H}_{1}=\left(V_{1},E_{1},f_{1}\right)$
and $\mathcal{H}_{2}=\left(V_{2},E_{2},f_{2}\right),$ a \textbf{homomorphism
of hypergraphs}\index{homomorphism!hypergraph}\index{hypergraph!homomorphism}
$h:\mathcal{H}_{1}\rightarrow\mathcal{H}_{2}$ consists of two functions,
abusively written $h$, $h:V_{1}\rightarrow V_{2}$ and $h:E_{1}\rightarrow E_{2}$
such that the diagram:

\begin{center}
\begin{tikzcd} 
E_1 \arrow[r, "f_1"] \arrow[d,"h"] 
& \mathcal{P}\left(V_{1}\right)\backslash\left\{ \emptyset\right\}  \arrow[d, "\mathcal{P}\left(h\right)"] \\ E_2 \arrow[r, "f_2"] 
& \mathcal{P}\left(V_{2}\right)\backslash\left\{ \emptyset\right\} 
\end{tikzcd}
\end{center}

is commutative.

We write $\mbox{\ensuremath{\boldsymbol{\mathcal{H}}\text{\textbf{yp}}}}$
the category whose objects are hypergraphs and whose morphisms are
hypergraph homomorphisms. We write $\mbox{\ensuremath{\boldsymbol{\mathcal{G}}\text{\textbf{raph}}}}$
the category of simple graphs and their homomorphisms.

The results of interest for this article are the followings:

\begin{lemma}

$\mbox{\ensuremath{\boldsymbol{\mathcal{G}}\text{\textbf{raph}}}}$
is a full subcategory of $\mbox{\ensuremath{\boldsymbol{\mathcal{H}}\text{\textbf{yp}}}}.$

\end{lemma}

\begin{lemma}

Duality of hypergraphs is a covariant functor $d:\mbox{\ensuremath{\boldsymbol{\mathcal{H}}\text{\textbf{yp}}}}\rightarrow\mbox{\ensuremath{\boldsymbol{\mathcal{H}}\text{\textbf{yp}}}}.$

\end{lemma}

\begin{lemma}

The incidence bipartite graph of a hypergraph allows a covariant functor
$B:\mbox{\ensuremath{\boldsymbol{\mathcal{H}}\text{\textbf{yp}}}}\rightarrow\mbox{\ensuremath{\boldsymbol{\mathcal{G}}\text{\textbf{raph}}}}.$

\end{lemma}

Given two hypergraphs $\mathcal{H}_{1}$ and $\mathcal{H}_{2}$ and
their respective incidence bipartite graph $\mathcal{B}_{\mathcal{H}_{1}}$
and $\mathcal{B}_{\mathcal{H}_{2}}$, and a homomorphism $h:\mathcal{H}_{1}\rightarrow\mathcal{H}_{2},$
$B\left(h\right):\mathcal{B}_{\mathcal{H}_{1}}\rightarrow\mathcal{B}_{\mathcal{H}_{2}}$
with $x\mapsto h(x)$ and $e\mapsto h(e).$

It is shown in \citet{dorfler1980categorytheoretical} that $B$ confuses
a hypergraph with its dual.

\begin{lemma}

The 2-section of a hypergraph induces the definition of a covariant
functor $G:\mbox{\ensuremath{\boldsymbol{\mathcal{H}}\text{\textbf{yp}}}}\rightarrow\mbox{\ensuremath{\boldsymbol{\mathcal{G}}\text{\textbf{raph}}}}.$

\end{lemma}

\begin{lemma}

To the intersection graph of a hypergraph corresponds the definition
of a covariant functor $I:\mbox{\ensuremath{\boldsymbol{\mathcal{H}}\text{\textbf{yp}}}}\rightarrow\mbox{\ensuremath{\boldsymbol{\mathcal{G}}\text{\textbf{raph}}}}.$

\end{lemma}

\section{Abstract simplicial complexes}

\label{subsec:Simplicial-complexes}

We just make a parenthesis on abstract simplicial complexes which
are in fact very particular hypergraphs.

In \citet{lee2010introduction}, the author introduces \textbf{abstract}
\textbf{simplicial complexes} as a collection $\mathcal{K}$ of nonempty
finite sets called abstract simplexes that is conditioned to only
one rule: if $S\in\mathcal{K}$, then any nonempty subset $S^{\prime}\subseteq S$,
$S^{\prime}\neq\emptyset$ is also in $\mathcal{K}.$ Such a subset
is called a face of $S$ and the elements of $S$ are called the \textbf{vertices}
of $S.$

Defining $V=\bigcup\limits _{S\in\mathcal{K}}S$ and for every $S\in\mathcal{K}$,
$E_{S}=\left(e_{S}\right)_{e_{S}\in\mathcal{P}\left(S\right)}$ and
$E=\underset{S\in\mathcal{K}}{\odot}E_{S}$ where $\odot$ is the
concatenation operation on families, i.e.\ the family of all the
elements of every family that is concatenated. $\mathcal{H}=\left(V,E\right)$
is then the unique hypergraph that corresponds to this simplicial
complex.

Reciprocally, given a hypergraph $\mathcal{H}=\left(V,E\right)$,
such that for every $e\in E$, if for all $e'\subseteq e$, $e'\in E$,
then $\mathcal{H}$ is in correspondence with a unique abstract simplicial
complex.

Simplicial complexes find their strength in the fact that simplicial
homology can be used on them, to study the number of holes of their
structure. Simplicial homology depends only on the associated topological
space and hence gives a way to distinguish simplicial complexes. The
interested reader can refer to \citet{hatcher2005algebraic} for an
introduction on simplicial homology\footnote{\href{http://pi.math.cornell.edu/~hatcher/AT/ATchapters.html}{Course of A. Hatcher on Cornell.edu}}.

As mentioned in \citet{wasserman1994social}, simplicial complexes
are useful to study the overlaps between subsets and the connectivity
of the network, as well as defining a dimensionality of the network.
More complex than hypergraphs by their hierarchical structure, they
have been used in different studies such as the structural analysis
of a game, using q-holes in \citet{gould1979structural}.

Abstract simplicial complexes are used quite often in the literature
for studying co-occurrence networks. In \citet{nikolaev2016modeling},
the author studies systems with higher-order interaction by considering
different models and particularly using a simplicial complex model
and a hypergraph model. It is also the case in \citet{salnikov2018cooccurrence}
where co-occurrence simplicial complexes are used to identify the
homological holes in mathematic knowledge, in order to discover emerging
knowledge. In \citet{benson2018simplicial}, the authors study the
temporal evolution of different co-occurrence datasets in order to
detect and predict higher-order interaction using simplicial closure,
studying particularly the life-cycles of triples and quadruples of
nodes, showing that the edge density and tie strength are two preponderant
factors in the evolution of the group to a closed simplex.

\section{Hypergraph coloring}

The subject of hypergraph coloring extends the one of graph coloring;
but there is not a single way of doing so. \citet{bujtas2015hypergraph}
is a full survey on the subject: we just give here a very basic definition
extracted from it and send the interested reader to this reference
for further details.

A \textbf{proper hypergraph coloring} is a mapping of a vertex of
$V$ to a color number taken in $\left\llbracket \lambda\right\rrbracket $
where $\lambda$ is a nonnegative integer such that for every hyperedge
there are at least two vertices of different colors. The minimum value
of $\lambda$ for which a proper coloring exists is called the \textbf{chromatic
number of the hypergraph} $\mathcal{H}$ and written $\chi\left(\mathcal{H}\right).$
A hypergraph is $k$-colorable if its chromatic number is no more
than $k$ and $k$-chromatic if $\chi\left(\mathcal{H}\right)=k.$

Hypergraph coloring has found applications in finding bounds of the
chromatic number of some graphs in \citet{kierstead1996applications}.
They are used in different optimization problems such as divide and
conquer approach in \citet{lu2004deterministic}. In \citet{voloshin1993mixed,voloshin2002coloring},
the author explores the coloring of mixed hypergraphs, in which the
hyperedge family is partitioned in two sub-families called the hyperedges
and anti-hyperedges. The author applies it to problem of energy supply.
Coloring is also used in partition problems in \citet{lu2004deterministic}.

The subject of hyperedge coloring is also tackled in \citet{gyarfas2013monochromatic}
to find monochromatic paths and cycles in hypergraphs.

\section{Hypergraph partitioning and clustering}

Hypergraph partitioning is defined in the Encyclopedia of Parallel
Computing\footnote{\href{https://link.springer.com/referenceworkentry/10.1007\%2F978-0-387-09766-4_1}{Hypergraph Partitioning in Encyclopedia of Parallel computing}}
as the ``task of dividing a hypergraph into two or more roughly equal-sized
parts such that a cost function on the hyperedges connecting vertices
in different parts is minimized.'' 

Very often this definition is too restrictive and requires strictly
more than two parts. The problem of partitioning a hypergraph is at
least NP-hard as mentioned in \citet{garey2002computers}. In \citet{karypis1999multilevel},
the authors present a hypergraph partitioning algorithm, called hMetis,
based on a multilevel coarsening of the hypergraph, working by iterative
bisections of the coarsened hypergraphs, starting from the smallest
one. In \citet{karypis2000multilevel}, the authors present an other
hypergraph partitioning that is $k$ ways, called hMeTiS-Kway algorithm. 

In \citet{papa2007hypergraph}, the author gives several methods for
hypergraph partitioning and see the action of clustering as ``The
process of computing a coarser hypergraph from an input hypergraph
by merging vertices into larger groups of vertices called clusters.''
Moreover, the author gives several applications of partitioning and
clustering, including VLSI design, numerical linear algebra, automated
theorem-proving and formal verification.

The literature is abundant in applications and methods: for more details
a survey on clustering ensembles techniques has been done in \citet{ghaemi2009survey},
which includes hypergraph partitioning.

Clustering and partitioning require often multi-level strategies approach:
this is a well studied domain, as it has been extensively used for
VLSI design---see for instance, \citet{karypis1999multilevel}---,
for parallel scientific computing---\citet{catalyurek1999hypergraphpartitioningbased},
\citet{devine2006parallel}, \citet{ballard2016hypergraph}---, for
image categorization---\citet{yuchihuang2011unsupervised}---, for
social networks---\citet{zhou2007learning}, \citet{yang2017hypergraph}. 

\section{Application of hypergraphs}

\label{subsec:Application-of-hypergraphs}

Hypergraphs fit to model multi-adic relationships in structures where
the traditional pairwise relationship of graphs is insufficient: they
are used in many areas such as social networks in particular in collaboration
networks---\citet{newman2001scientifica,newman2001scientific}---,
co-author networks---\citet{grossman1995portion}, \citet{taramasco2010academic}---,
chemical reactions---\citet{temkin1996chemical}---, genome\allowbreak---\citet{chauve2013hypergraph}---,
VLSI design---\citet{karypis1999multilevel}. Hypergraphs are also
used in information retrieval for different purposes such as query
formulation in text retrieval---\citet{bendersky2012modeling}---
or in music recommendation---\citet{bu2010music}. Several applications
of hypergraphs exist based on the diffusion process firstly developed
by \citet{zhou2007learning}. \citet{gao20123d} uses \citet{zhou2007learning}
for 3D-object retrieval and recognition by building multiple hypergraphs
of objects based on their 2D-views that are analyzed using the same
approach. In \citet{zhu2015contentbased}, multiple hypergraphs are
constructed to characterize the complex relations between landmark
images and are gathered into a multi-modal hypergraph that allows
the integration of heterogeneous sources providing content-based visual
landmark searches. Hypergraphs are also used in multi-feature indexing
to help image retrieval \citet{xu2016multifeature}. For each image,
a hyperedge gathers the first $n$ most similar images based on different
features. Hyperedges are weighted by average similarity. A spectral
clustering algorithm is then applied to divide the dataset into $k$
sub-hypergraphs. A random walk on these sub-hypergraphs allows to
retrieve significant images: they are used to build a new inverted
index, useful to query images. In \citet{wang2018joint}, a joint-hypergraph
learning is achieved for image retrieval, combining efficiently a
semantic hypergraph based on image tags with a visual hypergraph based
on image features.

\subsubsection{VLSI design}

Very Large Scale Integration---VLSI for short---aims at integrating
numerous transistors and devices into a single chip to create an integrated
circuit. In VLSI design, hypergraphs are intensively used for placement
of chipsets: in this case, the important thing is to achieve an efficient
clustering and partitioning of the hypergraph in order to minimize
connections between elements of the circuit; more details are given
in \citet{papa2007hypergraph}. In this case, the targeted layout
is an orthogonal layout. Dedicated algorithms have been implemented
to achieve this---\citet{eschbach2006orthogonal}, \citet{karypis1999multilevel,karypis2000multilevel}.

\subsubsection{Collaboration networks}

The comments of this Section are directly cited from \citet{ouvrard2017networks}.

As mentioned in \citet{newman2001scientific}, a collaboration is
a multi-adic relationship between occurrences, and, therefore, the
proper modeling is done by hypergraphs. Nonetheless, in \citet{newman2001scientific},
this multi-adic relationship is approximated by a 2-adic relationship
in between pairs of collaborators when it comes to be studied. The
same approximation is made in many other studies such as \citet{ramasco2004selforganization}.

The reasons for this approximation are numerous. It enables the use
of classical graph techniques and properties when studying the different
characteristics of collaboration networks, such as degree distribution,
clustering coefficient and also when applying quantifying metrics.
Today, many different techniques that help with the retrieval of information
from graphs are available.

This 2-adic relationship approximation has been developed in many
articles, where even if the multi-adic relationship of the data was
pointed to be more pertinent, this multi-adic relationship was not
used when getting to clustering. Since the end of the year 2000s,
the limitations of the 2-adic approach are more and more challenged,
as it leads to a partial loss of the information contained in the
multi-adic relationship. As a result, in \citet{estrada2005complex}
complex networks are modeled by hypergraphs.

In \citet{taramasco2010academic}, the authors study the academic
team formation using epistemic hypergraphs where hyperedges are subsets
of unions of a set of agents and a set of concepts. They introduce
new features to characterize the evolution of collaboration networks
taking into account the hypergraphic nature of networks. This paper
brings keystones in the study of a bi-dimensional hypergraph and shows
how the keeping of multi-adic relationships can help to gain in the
understanding of the evolution of a network.

\subsubsection{Recommender systems and hypergraphs}

Recommender systems contain two big methods of recommendation: the
item-item recommendation and the user-item recommendation. Both of
them require to make groups of items that have been bought for either
the purchase of the same item or for the same profile of user.

In \citet{bu2010music}, the authors use a unified hypergraph, i.e.\
a hypergraph that has multiple types of vertices to perform music
recommendation that takes not only into account the preference of
other users and similar music to the ones they listen. They start
by learning on the hypergraph using a cost function with a regularization
term that takes the vector of ranking scores as variable and the query
vector. The optimal solution for the ranking vector is found when
the gradient of the cost function is zero leading to an explicit expression
of the recommendation vector, based on the invert of the $\alpha$-Laplacian
matrix of the hypergraph, where $\alpha$ is learned off-line. In
\citet{lu2012recommender}, news recommendation is achieved by using
a hypergraph partition and a recommendation via ranking on the hypergraph
similar to the previous approach on a hybrid hypergraph that contains
seven different implicit relations with different objects.

In \citet{zhu2016heterogeneous}, a heterogeneous hypergraph is built
regrouping the information on users, tags and documents for document
recommendation. They consider different cost functions, one for annotation
relations, one for tag relations and one for document relations, based
on the Laplacian of the corresponding sub-hypergraphs. The final cost
relation is a linear combination of the three cost functions that
is optimized using the first $k$ eigenvectors corresponding to the
smallest eigenvalues of the total cost matrix.

In \citet{zheng2018novel}, a social network hybrid recommender system
is proposed based on hypergraph topological structure, that combined
with a hybrid matrix factorization model helps to enhance the description
of the interior relationships of a social network; in particular the
neighborhood of the users and items is used.

Additional references can be found to related work in the previous
references and in \citet{lu2012recommender}.

\subsubsection{Hypergraph grammar}

Hypergraph grammars are extension of graph grammars. Graph grammars,
also known as graph rewriting systems, aim at enumerating all the
possible graphs from a starting graph, by considering a set of graph
rewriting rules, searching an occurrence of the pattern graph (i.e.\
the current graph) and replacing it by an instance of the replacement
graph. Graph rewriting is intensively used in software engineering
for construction and verification, in layout algorithms, picture generation
but also in chemistry for the search of new molecules.

Several studies exist on hypergraph grammars. Particularly, in \citet{drewes1997hyperedge},
a handle hyperedge replacement is proposed: a handle is a sub-hypergraph
constituted of a hyperedge with its incident vertices, extending the
hyperedge replacement proposed in \citet{habel1987structural}.

Application of hypergraph grammars to molecular hypergraphs is still
an active subject of research---\citet{kajino2018molecular} for
instance. Hypergraph rewriting is also used for name binding in \citet{yasen2018name},
and also for graph design in \citet{luerssen2007graph}.

\subsubsection{Hypergraphs and linear algebra calculus}

In \citet{catalyurek1999hypergraphpartitioningbased}, the authors
use a hypergraph-partitioning of the rows (or columns) of a sparse
matrix that has to be multiplied by a vector to be used by an iterative
solver in order to achieve parallel computation. The partitioning
aims at minimizing the total communication volume required between
its different segments. The authors use two hypergraphs: one in which
the vertices represent the rows of the matrix and the hyperedges represent
the columns that have a nonzero intersection with the rows and the
other, its dual model.

There are several methods of matrix factorization that use hypergraph
regularization. They are all based on the fact that a hypergraph can
represent a matrix, by considering either the rows as vertices and
columns as hyperedges containing vertices of nonzero elements or reversely
using the dual hypergraph. This approach is applied to propose a hypergraph-based
non-symmetrical nested dissection ordering algorithm for LU decomposition
of sparse matrices in \citet{grigori2010hypergraphbased}.

In \citet{zeng2014image}, the authors propose a Hypergraph regularized
Non-negative Matrix Factorization (HNMF for short) to refine the classical
NMF approach and apply it to image clustering.

In \citet{jin2015lowrank}, the authors propose a multiple hypergraph
regularized low-rank matrix factorization, where the pool of hypergraphs
used for the regularization comes from the selection of neighbors
in the manifold through the bandwidth parameter used in the heat kernel.
The hypergraph Laplacian matrices are combined together allowing the
construction of a multi-hypergraph regularization that is added to
the regularization term of the original Truncated Singular Value Decomposition---TSVD
for short.

In \citet{wu2018nonnegative}, the authors propose a mixed hypergraph
regularized non-negative matrix factorization framework that constrains
the vertices of a hyperedge to be projected onto the same latent subspace.
The heterogeneity comes from the fact that hyperedges are built using
two kind of neighbors, one based on topological information and the
other on similarity information.

\section{Some additional comments}

A common objection to hypergraphs is: ``O.K., hypergraphs extend
graphs, but the same can be done with bipartite graphs.'' To answer
this objection different arguments can be developed. We regroup here
some of the common arguments found either on Internet\footnote{\href{https://cs.stackexchange.com/questions/12769/how-is-a-hypergraph-different-from-a-bipartite-graph}{https://cs.stackexchange.com/questions/12769/how-is-a-hypergraph-different-from-a-bipartite-graph}}\footnote{\href{http://en.wikipedia.org/wiki/Hypergraph\#Bipartite_graph_model}{http://en.wikipedia.org/wiki/Hypergraph\#Bipartite\_graph\_model}}
or in the literature.

First, graphs are particular case of hypergraphs: every graph is a
hypergraph, but not all hypergraphs are graphs. Bipartite graphs are
particular case of graphs: not every graph is a bipartite graph, but
all bipartite graphs are graphs. Hence: 

\[\circled{$\mathcal{B}$}\subsetneq\circled{$\mathcal{G}$}\subsetneq\circled{$\mathcal{H}$}\]

writing \circled{$\mathcal{B}$} the set of all bipartite graphs,
\circled{$\mathcal{G}$} the set of all graphs and \circled{$\mathcal{H}$}
the set of all hypergraphs.

Second, it is true that to a given hypergraph $\mathcal{H}=\left(V,E\right)$
corresponds a bipartite graph $G=\left(V\cup V',E_{G}\right).$ The
bipartite graph is obtained by considering for each hyperedge $e_{j}\in E$
an additional vertex $v_{e_{j}}$. $V'=\left\{ v_{e_{j}}:e_{j}\in E\right\} $
and that there is an edge between $v\in V$ and $v_{e_{j}}\in V'$
if and only if $v\in e_{j}.$ The bipartite graph is then called the
incidence graph of the hypergraph $\mathcal{H}$. Conversely, only
bipartite graphs with no isolated vertex in any of their vertex part
can represent a hypergraph with non-empty hyperedge, as required in
\citet{berge1967graphes,berge1973graphs}.

Third, from the previous argument, we can develop the argument of
factorization. A hypergraph is a factorized version of its incidence
graph: a hyperedge $e_{j}$ is defined as a collection of distinct
vertices: $e_{j}=\left\{ v_{j_{i}}:i\in\left\llbracket n_{j}\right\rrbracket \right\} ,$with
$n_{j}$ representing the cardinality of $e_{j}.$ The associated
bipartite graph, also called incidence graph, is the developed version:
there is an edge---i.e.\ a pair $\left(v_{j_{i}},e_{j}\right)$---per
vertex membership of hyperedge and the vertex set is extended to hold
a vertex representation of the original hyperedge. Hence, the vertex
set of the incidence graph is varying depending on the hyperedges
content, which is not the case for hypergraphs. Hypergraphs bring
the power of sets: operations on sets such as union, intersection,
complementation or subsets are well defined.

Nonetheless, one has to remark that a hypergraph and its dual have
the same corresponding bipartite graph: the bipartite graph confuses
the hypergraph with its dual \citet{dorfler1980categorytheoretical}.

Fourth, bipartite graphs correspond exactly to 2-colorable graphs.
Coloring a hypergraph is defined in \citet{berge1967graphes,berge1973graphs}
as assigning a color to each vertex such that no hyperedge, differing
from a singleton, is monochromatic. In \citet{seymour1974twocolouring},
a hypergraph that is not 2-colourable but where every proper subset
is 2-colourable is called a condenser. Condensers are shown to be
minimal hypergraphs that are not 2-colourable. Examples of such condensers
are: $\left\{ \left\{ 1,2\right\} ,\left\{ 2,3\right\} ,\ldots,\left\{ n-1,n\right\} ,\left\{ n,1\right\} \right\} $
with $n$ odd or $\left\{ \left\{ 1,2,\ldots,n\right\} \right\} \cup\left\{ \left\{ 0,i\right\} :1\leqslant i\leqslant n\right\} $
with $n\geqslant2.$

Fifth, in \citet{zykov1974hypergraphs}, the author insists on the
fact that while the isomorphism between a hypergraph and its incident
bipartite representation, called the König graph, is indubitable,
some of the statements made with hypergraphs have a very ``artificial
and cumbersome character in terms of König representation which obscures
the point''. The author takes the example of the chromatic number
definition of an ordinary graph, which is a very particular case of
hypergraph, in terms of its König representation.

Sixth, in \citet{berge1967graphes,berge1973graphs}, the author focuses
on some problems that have solutions only with hypergraphs: the representative
graph or line graph is the 2-section of the dual hypergraph. Given
a representative graph $G,$ the question raised is to know if there
exists a hypergraph $\mathcal{H}$ having $G$ as its representative
graph. The problem is shown feasible for general hypergraphs, but
\citet{berge1967graphes,berge1973graphs} shows that this problem
has no solution when it is required that $\mathcal{H}$ is $r$-uniform
and that a vertex belongs to more than $r$ cliques. Hence, the graph
shown in Figure \ref{Fig:Counter_example_representative_graph} cannot
be the representative graph of a $2$-uniform hypergraph.

\begin{figure}
\begin{center}

\definecolor{lavander}{cmyk}{0,0.48,0,0} \definecolor{violet}{cmyk}{0.79,0.88,0,0} \definecolor{burntorange}{cmyk}{0,0.52,1,0}
\def\lav{lavander!90} \def\oran{orange!30}
\tikzstyle{superpeers}=[draw,circle,burntorange, left color=\oran,                        text=violet,minimum width=10pt]
\begin{tikzpicture}[auto, thick]   
\foreach \place/\name in {{(0,0)/a}, {(1,1.732)/b}, {(-1,1.732)/c}, {(-2,0)/d},            {(-1,-1.732)/e},{(1,-1.732)/f},{(2,0)/g}}     
	\node[superpeers] (\name) at \place {$\name$};
\foreach \source/\target in {a/b,a/c,b/c,a/d,a/e,d/e,a/f,a/g,f/g}    
	\path (\source) edge (\target);
\end{tikzpicture}

\end{center}

\caption{Counter-example of graph that cannot be the representative graph of
a $2$-uniform hypergraph.}
\label{Fig:Counter_example_representative_graph}
\end{figure}
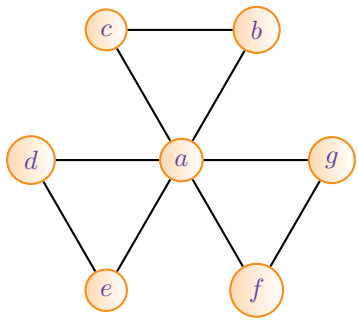

\section{Acknowledgments}

This work was done during the writing of the Thesis of Xavier OUVRARD,
done at University of Geneva, supervised by Pr. Stéphane MARCHAND-MAILLET
and founded by a doctoral position at CERN, in Collaboration Spotting
team, supervised by Dr. Jean-Marie LE GOFF.

\bibliographystyle{plainnat}

\end{document}